\documentclass[authoryear,3p,preprint,12pt]{elsarticle}

\usepackage{comment}
\usepackage[utf8]{inputenc}
\usepackage{amsmath}
\usepackage[english]{babel}
\usepackage{amssymb}
\usepackage{amsfonts}
\usepackage{bm}
\usepackage{graphicx}
\usepackage{amsmath}
\usepackage{amsthm}
\usepackage{xcolor}
\usepackage{parskip}
\usepackage{multicol}
\usepackage{float}
\usepackage{rotating}
\usepackage{subcaption}
\usepackage{verbatim}
\usepackage{chngcntr}
\usepackage{apptools}
\usepackage{rotating}
\usepackage{tabularx}
\usepackage{booktabs, makecell, multirow}
\usepackage{setspace}
\usepackage[tableposition=top]{caption}
\usepackage{multirow}
\usepackage{lscape}
\usepackage{blindtext}
\usepackage{booktabs}
\usepackage{animate}
\usepackage{lipsum}
\usepackage{natbib}
\usepackage{adjustbox}

\theoremstyle{definition}
\newtheorem{theorem}{Theorem}
\newtheorem{cor}{Corollary}
\newtheorem{lem}{Lemma}

\newtheorem{remark}{Remark}

\usepackage{natbib}
\usepackage[colorlinks=true,linkcolor=black, citecolor=blue, urlcolor=blue]{hyperref}

\usepackage{amssymb}

 \usepackage{lineno}

\journal{ }

\begin{document}


\begin{frontmatter}


\title{Flexible Validity Conditions for the Multivariate Mat\'ern Covariance in any Spatial Dimension and for any Number of Components}

\author[1,2]{Xavier Emery}
\address[1]{Department of Mining Engineering, University of Chile, Avenida Beauchef 850, Santiago, Chile.}
\address[2]{Advanced Mining Technology Center, University of Chile, Avenida Beauchef 850, Santiago, Chile.}
\author[3,4]{Emilio Porcu}
\address[3]{Department of Mathematics, Khalifa University of Science and Technology, Abu Dhabi, The United Arab Emirates.}
\address[4]{School of Computer Science and Statistics, Trinity College, Dublin, Ireland.}
\author[5]{Philip White}
\address[5]{Department of Statistics, Brigham Young University, Provo, UT, USA.}


\begin{abstract}

This paper addresses the problem of finding parametric constraints that ensure the validity of the multivariate Mat{\'e}rn covariance for modeling the spatial correlation structure of coregionalized variables defined in an Euclidean space. To date, much attention has been given to the bivariate setting, while the multivariate setting has been explored to a limited extent only. The existing conditions often imply severe restrictions on the upper bounds for the collocated correlation coefficients, which makes the multivariate Mat{\'e}rn model appealing for the case of weak spatial cross-dependence only. We provide a collection of sufficient validity conditions for the multivariate Mat{\'e}rn covariance that allows for more flexible parameterizations than those currently available, and prove that one can attain considerably higher upper bounds for the collocated correlation coefficients in comparison with our competitors. We conclude with an illustration on a trivariate geochemical data set and show that our enlarged parametric space yields better fitting performances.

\end{abstract}


\end{frontmatter}

\section{Introduction}

\subsection{Motivation and Context}

This work is motivated by highly-multivariate spatial modeling problems in geosciences and mining engineering applications. For instance, in exploration geochemistry, it is of interest to map the concentrations of several tens to hundreds of elements, with the object of detecting concealed mineral deposits \citep{castillo, Goovaerts, guartan2, guartan}.
In mineral resources evaluation, the grades of several elements of interest (main products, by-products, and contaminants) often need to be jointly predicted or simulated \citep{mery, Minniakhmetov}.
In geotechnics, the rock mass rating classification system is obtained as the combination of several basic parameters (uniaxial compressive strength, rock quality designation, spacing of discontinuities, condition of discontinuities, groundwater conditions, and orientation of discontinuities) \citep{pinheiro}.
In geometallurgy, one deals with information on mineral proportions, grain sizes, rock density, texture, indices of fragmentation, abrasion, crushing or grinding of the rock, mass recoveries, and metallurgical recoveries, leading to tens or hundreds of variables to be jointly modeled \citep{boisvert, boogaart}.
There has been considerable effort in the recent literature to find flexible multivariate spatial correlation models that overcome the drawbacks of the well-known linear model of coregionalization (LMC).
Yet, the LMC is often the default approach for multivariate spatial modeling within the above branches of applied sciences.

In the aforementioned applications, the coregionalized variables of interest are often viewed as realizations of a vector random field defined in an Euclidean space,  $\bm{Z}=\{\bm{Z}(\bm{s})=(Z_1(\bm{s}),\ldots,Z_p(\bm{s}))^\top: \bm{s}\in \mathbb{R}^d \}$, with $p$ standing for the number of variables and $d$ for the space dimension.
$\bm{Z}$ is second-order stationary if its expected value exists and is a constant vector in $\mathbb{R}^d$, and if the covariance between $\bm{Z}(\boldsymbol{s})$ and $\bm{Z}(\boldsymbol{s}^\prime)$ depends exclusively on $\boldsymbol{h}=\boldsymbol{s}-\boldsymbol{s}^\prime$. This stationarity assumption is common in order to facilitate inference on the spatial correlation structure of the random field from a set of scattered data.
A stationary covariance function  ${\bf K}: \mathbb{R}^d \rightarrow \mathbb{R}^{p \times p}$ is a positive semidefinite matrix-valued mapping whose elements are defined as ${K}_{ij}(\bm{h}) = \text{cov}\{Z_i(\bm{s}),Z_j(\bm{s}+\bm{h})\}$, $i,j=1,\cdots,p$, and $\bm{s}, \bm{h} \in \mathbb{R}^d$.
Showing positive semidefiniteness of any candidate ${\bf K}$ is often a challenging task, see, for instance, \cite{Wackernagel:2003}, \cite{Chiles2012}, \cite{marcotte2019} and references therein.

\cite{Gneiting:Kleibler:Schlather:2010} introduced the multivariate Mat{\'e}rn coregionalization model, for which each direct or cross-covariance is a univariate Mat{\'e}rn covariance function. The latter univariate covariance has been extremely popular in spatial statistics because it possesses a smoothness parameter that allow parameterizing the short-scale regularity of the underlying random field, and because it has a simple form for the related spectral density, hence it is an appealing model from the perspective of spectral modeling and simulation \citep{Lantu2002, emery2016, lauzon}.
The main focus in \cite{Gneiting:Kleibler:Schlather:2010} is on necessary and sufficient validity conditions ({\em i.e.}, positive semidefiniteness) on the parameters indexing the multivariate covariance function for the bivariate case ($p=2$). \cite{Apanasovich} provided sufficient validity conditions for the $p$-variate Mat{\'e}rn model with $p>2$.

\subsection{Pros and Cons of Multivariate Mat{\'e}rn Modeling: State of the Art}
\label{prosandcons}

The multivariate Mat{\'e}rn model is a breakthrough in the literature where the LMC has been central for decades \citep{Journel, Goulard, emery2010, Bailey2012, Desassis2013}. The LMC represents any component of a $p$-variate random field $\boldsymbol{Z}$ as a linear combination of latent, uncorrelated scalar fields.  \cite{Gneiting:Kleibler:Schlather:2010} and \cite{daley2015} advocate for different modeling strategies because of the potential drawbacks of the LMC: the smoothness of any component of the $p$-variate field amounts to that of the roughest latent scalar field.  Moreover, the number of parameters can quickly become huge as the number $p$ of components increases. \\
The multivariate Mat{\'e}rn model is more flexible in that it allows for different smoothness for the components of the $p$-variate random field. Yet, such flexibility is limited by restrictions to the parameter space that ensure ${\bf K}$ is positive semidefinite. Notably, in the bivariate case, one of the most important restrictions amounts to an upper bound for the absolute value of the collocated correlation coefficient, $\rho_{12}$.
For $p>2$ the commensurate restrictions on the $p(p-1)/2$ collocated correlation coefficients become overly severe. \\
While the bivariate case has been characterized completely by \cite{Gneiting:Kleibler:Schlather:2010}, only sufficient conditions are available for the case $p>2$ \citep{Gneiting:Kleibler:Schlather:2010, Apanasovich, Du2012}. Further, \cite{Gneiting:Kleibler:Schlather:2010} and \cite{Du2012} work under the very restrictive assumption that all the marginal and cross scale parameters are equal. \cite{Apanasovich} improve on these conditions by providing new constraints that are dependent on the dimension $d$ of the space where the random field is defined.
Yet, these sets of conditions still imply severe restrictions in terms of upper bound for the  $p(p-1)/2$ collocated correlation coefficients. This can result in an unrealistic model that works under the assumption of quasi-independence between the components of the $p$-variate random field, as such upper bounds are often close to zero. We explore this aspect in Section~\ref{sec5}.

\subsection{Our Contribution}

The objective of this paper is to improve the definition of the validity region for the parameters of the multivariate Mat\'ern covariance model, by proposing new sufficient conditions that are broader or less restrictive than the currently known ones. We do not address any modification in the model itself or its use in estimation, prediction or simulation studies, for which we refer to the existing literature \citep{Gneiting:Kleibler:Schlather:2010, Apanasovich}. In particular, examining the model complexity, comparing estimation approaches or analyzing the impact of the model parameters in prediction outputs are out of the scope of this work. The interest is to find novel conditions under which the $p$-variate Mat\'ern model is positive semidefinite in $\mathbb{R}^d$ for any given spatial dimension $d$ and number of random field components $p$.

Table~\ref{resume} gives a comparative sketch of \cite{Gneiting:Kleibler:Schlather:2010},  \cite{Apanasovich}, \cite{Du2012}, and the findings provided in this paper. The table omits the case $p=2$, insofar as \cite{Gneiting:Kleibler:Schlather:2010} identified necessary and sufficient conditions for this case. For the case $p>2$, we give several sufficient conditions for the validity of the multivariate Mat{\'e}rn model, with some specific cases (last column), as well as more general cases (second and third columns). \\

\begin{landscape}
\tiny{
\newcolumntype{M}[1]{>{\centering\arraybackslash}m{#1}}
\begin{table}[ht]
  \centering
 \scalebox{0.8}{
    \begin{tabular}{p{20mm}p{55mm}p{75mm}p{70mm}} \toprule
        { Contribution by}   & {General conditions depending on a matrix-valued hyperparameter} $\boldsymbol{\psi}$   & { General conditions depending on} ($\boldsymbol{\nu},\boldsymbol{\alpha},\boldsymbol{\sigma}$) {and, possibly, a scalar hyperparameter} $\beta>0$ or $\delta \geq 0$ & { Specific conditions depending on} ($\boldsymbol{\nu},\boldsymbol{\alpha},\boldsymbol{\sigma}$) { with either } $\boldsymbol{\nu}=\nu \boldsymbol{1}$ { and } $\nu>0${, or } $\boldsymbol{\alpha}=\alpha \boldsymbol{1}$ { and } $\alpha>0$  \\ \midrule
       This paper $\qquad \qquad \qquad$  &   &  \begin{minipage}{0.9\textwidth} {\small
       Theorem~\ref{validitymatern20}A\\
       $\boldsymbol{\nu} \in {\cal B}(\boldsymbol{s}_1,\ldots,\boldsymbol{s}_p)$ \\ $\boldsymbol{\alpha}^{-2} \in{\cal B}(\boldsymbol{s}_1,\ldots,\boldsymbol{s}_p)$ \\ $\boldsymbol{\sigma} \, \boldsymbol{\alpha}^{-d} \, \Gamma(\boldsymbol{\nu}+d/2)/\Gamma(\boldsymbol{\nu}) \ge 0$.}
       \end{minipage} &  \\ \cmidrule{2-4}

       & \begin{minipage}{0.9\textwidth} {\small
       Theorem~\ref{validitymatern20}B\\
       $\boldsymbol{\psi} \le 0$ \\ $\boldsymbol{\nu} \le 0$\\
       $\boldsymbol{\alpha}^2 \, \boldsymbol{\psi} -\boldsymbol{\nu} \le 0  $\\
        $\boldsymbol{\sigma} \,  \boldsymbol{\psi}^{\boldsymbol{\nu} +d/2} \, \boldsymbol{\alpha}^{2 \boldsymbol{\nu} } \exp \left ( - \boldsymbol{\nu} \right  )/\Gamma(\boldsymbol{\nu} ) \ge 0   $.}
       \end{minipage} &  \begin{minipage}{0.9\textwidth} {\small
       Theorem~\ref{validityMatern}A\\
       $\boldsymbol{\nu} \le 0 $ \\ $\boldsymbol{\nu} \, \boldsymbol{\alpha}^{-2}  \le 0$ \\
       $\boldsymbol{\sigma} \, \boldsymbol{\alpha}^{-d} \, \boldsymbol{\nu}^{\boldsymbol{\nu}+d/2} \exp(-\boldsymbol{\nu})/\Gamma(\boldsymbol{\nu} ) \ge 0$.}
       \end{minipage} & \begin{minipage}{0.9\textwidth} {\small
       Example 1\\
       $\boldsymbol{\nu} =\nu \boldsymbol{1}$ \\ $\boldsymbol{\alpha}^{-2} \le 0$ \\ $\boldsymbol{\sigma} \, \boldsymbol{\alpha}^{-d}  \ge 0$.}
       \end{minipage} \\ \cmidrule{4-4}
         &      &   &  \begin{minipage}{0.9\textwidth} {\small
       Example 2\\
       $\boldsymbol{\alpha} =\alpha \boldsymbol{1}$ \\ $\boldsymbol{\nu} \le 0$ \\ $\boldsymbol{\sigma} \, \boldsymbol{\nu}^{\boldsymbol{\nu}+d/2} \exp(-\boldsymbol{\nu})/\Gamma(\boldsymbol{\nu} ) \ge 0$.}
       \end{minipage} \\ \cmidrule(l){3-4}
               & &  \begin{minipage}{0.9\textwidth} {\small
       Theorem~\ref{validityMatern}B\\
       $\boldsymbol{\nu} \le 0 $ \\ $\boldsymbol{\alpha}^{2} - \beta \boldsymbol{\nu} \le 0$ \\ $\boldsymbol{\sigma} \left ( \boldsymbol{\alpha}^2/\beta \right )^{\boldsymbol{\nu}} \exp(-\boldsymbol{\nu}) /\Gamma(\boldsymbol{\nu} ) \ge 0$.}
       \end{minipage} & \begin{minipage}{0.9\textwidth} {\small
       Example 3\\
       $\boldsymbol{\nu} =\nu \boldsymbol{1}$ \\ $\boldsymbol{\alpha}^{2} \le 0$ \\ $\boldsymbol{\sigma} \, \boldsymbol{\alpha}^{2 \boldsymbol{\nu} }  \ge 0$.}
       \end{minipage} \\ \cmidrule(l){2-4}
                & &  &  \begin{minipage}{0.9\textwidth} {\small
       Theorem~\ref{validityexponentialdD}\\
       $\boldsymbol{\nu} =\nu \boldsymbol{1}$ \\ $\boldsymbol{\alpha} \le 0$ \\ $\boldsymbol{\sigma} \, \boldsymbol{\alpha}^{\lfloor (d+1+3 \lceil 2 \nu  \rceil)/2 \rfloor}  \ge 0$.}
       \end{minipage} \\ \cmidrule(l){1-4}
        {\bf \cite{Du2012}}         & &  &  \begin{minipage}{0.9\textwidth} {\small
       $\boldsymbol{\alpha} =\alpha \boldsymbol{1}$ \\ $\boldsymbol{\nu} \le 0$ \\ $\boldsymbol{\sigma} \, \Gamma(\boldsymbol{\nu}+d/2)/\Gamma(\boldsymbol{\nu})  \ge 0$.}
       \end{minipage} \\ \cmidrule(l){1-4}
      {\bf \cite{Apanasovich}} &  & \begin{minipage}{0.9\textwidth} {\small
       $\nu_{ij}-\frac{\nu_{ii}+\nu_{jj}}{2} = \delta (1-a_{ij})$ {with} $a_{ii}=1, a_{ij}\in [0,1]$ \\
       $[a_{ij}]_{i,j} \ge 0$ \\
       $\boldsymbol{\alpha}^2 \le 0$ \\ $\left[\frac{\sigma_{ij} \, {\alpha_{ij}}^{2 \delta + \nu_{ii} +\nu_{jj} } \, \Gamma(\nu_{ij}+d/2)} {\Gamma(\nu_{ii}/2+\nu_{jj}/2+d/2) \, \Gamma(\nu_{ij})}\right]_{i,j} \ge 0$.}
       \end{minipage} &  \\ \cmidrule(l){1-4}
       {\bf \cite{Gneiting:Kleibler:Schlather:2010} }       &  &  & \begin{minipage}{0.9\textwidth} {\small
       $\nu_{ij}=(\nu_{ii}+\nu_{jj})/2 $ \\ $\boldsymbol{\alpha} = \alpha \boldsymbol{1}$ \\ $\left[\frac{\sigma_{ij} \, \Gamma(\nu_{ii})^{1/2} \, \Gamma(\nu_{jj})^{1/2} \, \Gamma((\nu_{ii}+\nu_{jj}+d)/2)}   {\Gamma(\nu_{ii}+d/2)^{1/2} \, \Gamma(\nu_{jj}+d/2)^{1/2} \, \Gamma((\nu_{ii}+\nu_{jj})/2)} \right]_{i,j} \ge 0$.}
       \end{minipage} \\ \bottomrule
    \end{tabular}}
    \caption{A map of the results contained in this paper, with comparison with earlier literature. Here, $\le 0$ stands for {\em is conditionally negative semidefinite}, $\ge 0$ for {\em is positive semidefinite} and $\in {\cal B}(\boldsymbol{s}_1,\ldots,\boldsymbol{s}_p)$ for {\em is a Bernstein matrix with supporting points $\boldsymbol{s}_1,\ldots,\boldsymbol{s}_p$}. All the matrix operations (product, division, inverse, power and exponential) are taken element-wise. The other symbols are clarified in the paper. \label{resume} }
\end{table}
}
\end{landscape}

The remainder of the paper is as follows: Section~\ref{sec2} contains the necessary mathematical background. Section~\ref{sec3} presents new conditions for parsimonious and general parameterizations of the multivariate Mat\'ern model.
Section~\ref{sec5} provides a comparison with \cite{Apanasovich} of the upper bounds for the collocated correlation coefficient in three specific examples. Section~\ref{sec6} illustrates our findings through a geochemical data set of three coregionalized variables. Section~\ref{sec7} concludes the paper with a short discussion. The proofs of the technical results are deferred to Appendices~\ref{app1} to~\ref{app3}, where we also include some other background material and lemmas.

\section{Background} \label{sec2}
\label{Background}

\subsection{The Mat{\'e}rn Class of Correlation Functions}
\label{Matern-univ}
The isotropic Mat{\'e}rn correlation function in $\mathbb{R}^d$, $d \geq 1$,
is given by
\citep{Matern}
\begin{align}
\label{matern}
k(\boldsymbol{h} ; \alpha,\nu) = \frac{2^{1-\nu}}{\Gamma(\nu)} (\alpha \|\boldsymbol{h}\| )^\nu {\cal K}_\nu(\alpha \|\boldsymbol{h}\| ), \quad \boldsymbol{h} \in \mathbb{R}^d,
\end{align} with $\|\cdot\|$ the Euclidean norm,
$\Gamma$  the gamma function, ${\cal K}_\nu$ the modified Bessel function of the second kind,
$\nu$ and $\alpha$ positive smoothness and scale parameters, respectively.
 The parameter $\nu$ indexes both the mean square differentiability and the fractal dimension of a Gaussian random field having a Mat{\'e}rn correlation function.

Being a correlation function and integrable in $\mathbb{R}^d$ for all positive $\nu$, the function $k(\cdot;\alpha,\nu)$ is the Fourier transform of a positive and bounded measure that is absolutely continuous with respect to the Lebesgue measure:
\begin{align}
\label{spectralmatern}
k(\boldsymbol{h} ; \alpha,\nu) = \int_{\mathbb{R}^d} \cos(\langle \boldsymbol{h}, \boldsymbol{\omega} \rangle) \, \widetilde{k}(\boldsymbol{\omega}; \alpha,\nu) \text{d}\boldsymbol{\omega}, \quad \boldsymbol{h} \in \mathbb{R}^d,
\end{align}
with $\langle \cdot,\cdot \rangle$ the inner product and $\widetilde{k}$ the isotropic spectral density of the Mat\'ern covariance, defined as \citep{Lantu2002}
\begin{align}
\label{spectraldensitymatern}
 \widetilde{k}(\boldsymbol{\omega}; \alpha,\nu) = \frac{\Gamma(\nu+d/2)}{\Gamma(\nu)\, \alpha^d \, \pi^{d/2}} \left(1+\frac{\|\boldsymbol{\omega}\|^2}{\alpha^2}\right)^{-\nu-d/2}, \quad \boldsymbol{\omega} \in \mathbb{R}^d.
\end{align}
We refer to \eqref{spectraldensitymatern} as the spectral representation of $k(\cdot;\alpha,\nu)$. Another relevant fact is that the mapping $t \mapsto k(\sqrt{t};\alpha,\nu)$ is completely monotonic on the positive real line. Hence,
the mapping \eqref{matern} can be written as a mixture of Gaussian covariances of the form $g(\boldsymbol{h};u)=\exp(-u \|\boldsymbol{h}\|^2)$, where $1/u$ follows a gamma distribution with shape parameter $\nu$ and scale parameter $\alpha^2/4$ \citep[formula 3.471.9]{EmeryLantu2006, Grad}:
 \begin{align}
  \label{integral_repr}
 k(\boldsymbol{h}; \alpha,\nu) =  \int_{0}^{+\infty}   g(\boldsymbol{h};u) f(u; \alpha,\nu) \text{d}u, \qquad \boldsymbol{h} \in \mathbb{R}^d,
  \end{align}
where $$f(u;\alpha,\nu) =  \frac{1}{\Gamma(\nu)} \left( \frac{\alpha}{2} \right)^{2\nu} u^{-\nu-1} \exp\left(-\frac{\alpha^2}{4u}\right), \qquad u > 0, $$
 is the inverse gamma probability density function with parameters $\nu$ and  $\alpha^2/4$.

If $\nu$ is an half-odd-integer, $k(\cdot ; \alpha,\nu)$ separates into the product of a negative exponential function with a polynomial of degree $\nu-1/2$. For instance, $\nu=1/2$ and $\nu=3/2$ correspond to $k(\boldsymbol{h} ; \alpha,1/2)= \exp(-\alpha \|\boldsymbol{h}\|)$ and $k(\boldsymbol{h} ; \alpha,3/2)= \exp(-\alpha \|\boldsymbol{h}\|)(1+\alpha \|\boldsymbol{h}\| )$, which are associated with Gaussian random fields being the continuous versions of a Ornstein-Uhlenbeck process and a first-order autoregressive process, respectively \citep{BFFP}.
Furthermore, with a suitable rescaling, the Mat\'ern correlation function tends to a Gaussian correlation function:
\begin{equation}
\label{convergencetogaussian}
    k(\boldsymbol{h};2 \sqrt{\beta \nu},\nu) \xrightarrow[\nu \to +\infty]{} g(\boldsymbol{h};\beta) = \exp(-\beta \|\boldsymbol{h}\|^2),
\end{equation} with a uniform convergence on any compact set of $\mathbb{R}^d$. This convergence is established by showing that the spectral density of $k(\cdot;2\sqrt{\beta \nu},\nu)$ tends pointwise to that of $g(\cdot;\beta)$ \citep{Lantu2002} as $\nu$ tends to infinity:
\begin{align}
\label{spectraldensitygauss}
 \widetilde{k}(\boldsymbol{\omega}; 2\sqrt{\beta \nu},\nu) \xrightarrow[\nu \to +\infty]{} \widetilde{g}(\boldsymbol{\omega};\beta) = \frac{1}{(4 \pi \beta)^{d/2}} \exp\left(-\frac{\|\boldsymbol{\omega}\|^2}{4 \beta} \right), \quad \boldsymbol{\omega} \in \mathbb{R}^d.
\end{align}

\subsection{Multivariate Mat{\'e}rn and Gaussian Models}
\label{Matern-multi}

The $p$-variate isotropic Mat\'ern covariance model in $\mathbb{R}^d$ is defined as \citep{Gneiting:Kleibler:Schlather:2010}
\begin{equation}
\label{multi_cov}
    {\bf K}({\boldsymbol{h}};\boldsymbol{\alpha},\boldsymbol{\nu},\boldsymbol{\sigma}) = [\sigma_{ij} \, k(\boldsymbol{h};\alpha_{ij},\nu_{ij})]_{i,j=1}^p, \quad \boldsymbol{h} \in \mathbb{R}^d,
\end{equation}
where $\boldsymbol{\alpha} = [\alpha_{ij}]_{i,j=1}^p$ and $\boldsymbol{\nu}=[\nu_{ij}]_{i,j=1}^p$ are symmetric matrices with positive entries (scale and smoothness parameters, respectively), while $\boldsymbol{\sigma} = [\sigma_{ij}]_{i,j=1}^p$ is a symmetric matrix (collocated covariance) with entries in $\mathbb{R}$. As a particular case, if $\boldsymbol{\nu}=\boldsymbol{1}/2$, with $\boldsymbol{1}$ being a matrix of ones, one obtains the $p$-variate exponential covariance. A direct application of \eqref{integral_repr} shows that
\begin{align}
\label{spectral_biv}
{\bf K}({\boldsymbol{h}};\boldsymbol{\alpha},\boldsymbol{\nu},\boldsymbol{\sigma}) = \int_{0}^{+\infty}   g(\boldsymbol{h};u) \boldsymbol{\sigma} \, {\bf F}(u; \boldsymbol{\alpha},\boldsymbol{\nu}) \text{d}u,
\end{align}
with ${\bf F}(u; \boldsymbol{\alpha},\boldsymbol{\nu}) = [f(u; \alpha_{ij},\nu_{ij})]_{i,j=1}^p$. Here, the product and integration are intended as element-wise.

 A necessary and sufficient validity condition for ${\bf K}$ in \eqref{multi_cov} is that the spectral density matrix $\widetilde{{\bf K}}({\boldsymbol{\omega}};\boldsymbol{\alpha},\boldsymbol{\nu},\boldsymbol{\sigma})$ with entries $ \widetilde{K}_{ij}(\boldsymbol{\omega})=\sigma_{ij} \widetilde{k}(\boldsymbol{\omega}; \alpha_{ij},\nu_{ij})$ and $\widetilde{k}$ as defined in \eqref{spectraldensitymatern}, is  a positive semidefinite matrix  for every fixed $\boldsymbol{\omega} \in \mathbb{R}^d$ \citep[][with the references therein]{AlonsoMalaver2015251}. This result has been used to obtain necessary and sufficient validity conditions for the bivariate Mat{\'e}rn model \cite[theorem 3]{Gneiting:Kleibler:Schlather:2010}. Yet, applying these validity conditions becomes prohibitive when $p>2$.

The $p$-variate Gaussian covariance model is defined as
\begin{equation}
  \label{multi_gauss}
   {\bf G}({\boldsymbol{h}};\boldsymbol{\beta},\boldsymbol{\sigma}) = [\sigma_{ij}\, g(\boldsymbol{h};\beta_{ij})], \quad \boldsymbol{h} \in \mathbb{R}^d,
\end{equation}
with $\boldsymbol{\beta} = [\beta_{ij}]_{i,j=1}^p$ a symmetric matrix with positive entries and $\boldsymbol{\sigma} = [\sigma_{ij}]_{i,j=1}^p$ a symmetric matrix with entries in $\mathbb{R}$.
This model is the limit of a properly rescaled multivariate Mat\'ern covariance \eqref{multi_cov} when the smoothness parameters $\nu_{ij}$ tend to infinity, as per \eqref{convergencetogaussian}.

\subsection{Conditionally Negative Semidefinite and Bernstein Matrices}
A $p \times p$ symmetric real-valued matrix $\boldsymbol{A}=[a_{ij}]_{i,j=1}^p$ is said to be conditionally negative semidefinite if, for any vector $\boldsymbol{\lambda}=(\lambda_1,\ldots,\lambda_p)$ in $\mathbb{R}^p$ whose components add to zero, one has $\sum_{i=1}^p \sum_{j=1}^p \lambda_i \, \lambda_j \, a_{ij} \leq 0$. Checking whether or not a symmetric matrix is conditionally negative semidefinite is straightforward, as it amounts to checking the positive semidefiniteness of a related matrix (see lemma~\ref{superlemma} in Appendix~\ref{app1}).

As an example, for a positive integer $d$, a set of real numbers $\{ \eta_{1}, \ldots, \eta_{p} \}$, a set of points in $\mathbb{R}^d$  $\{ \boldsymbol{s}_1, \ldots, \boldsymbol{s}_p \}$, and a variogram function $\gamma: \mathbb{R}^d \times \mathbb{R}^d \to \mathbb{R}_+$, the matrix $[\eta_{ij}]_{i,j=1}^p$, with generic entry
\begin{equation}
\label{ejemplo}
\eta_{ij} = \frac{\eta_{i}+\eta_{j}}{2} + \gamma(\boldsymbol{s}_i,\boldsymbol{s}_j),
\end{equation}
is conditionally negative semidefinite. Recall that a variogram $\gamma: \mathbb{R}^d \times \mathbb{R}^d \to \mathbb{R}$ is a mapping that represents half the variance of the increments of a random field defined in $\mathbb{R}^d$ \citep{Matheron1965}. If $\gamma(\boldsymbol{s}_1,\boldsymbol{s}_2) \equiv \tilde{\gamma}(\boldsymbol{s}_2- \boldsymbol{s}_1)$, the variogram is said to be intrinsically stationary.

Intrinsically stationary variograms can be constructed by means of primitives of completely monotonic functions. A continuous function $b:[0,+\infty) \to \mathbb{R}$ such that $b(0)<+\infty$ is said to be completely monotonic if it is infinitely differentiable and $(-1)^n b^{(n)} \ge 0$, for all $n \in \mathbb{N}$, and where $b^{(n)}$ means the $n$-th derivative of $b$. Here, we abuse of notation when writing $b^{(0)}$ for $b$. A nonnegative and continuous function $B: [0,+\infty) \to \mathbb{R}$ is called a Bernstein function \citep{Schilling, porcu2011schoenberg} if its first derivative is completely monotonic. Arguments in \cite{porcu2011schoenberg} show that, for any Bernstein function $B$, the mapping
\begin{equation} \label{vario}
    \boldsymbol{h} \mapsto \tilde{\gamma}(\boldsymbol{h})=B(\|\boldsymbol{h} \|^\theta), \quad 0 < \theta \leq 2,
\end{equation}
 is an intrinsically stationary variogram in $\mathbb{R}^d$ for any positive integer $d$. In particular, for $\theta=1$, the matrix $\boldsymbol{B}=[B(\|\boldsymbol{s}_i-\boldsymbol{s}_j\|)]_{i,j=1}^p$ is conditionally negative semidefinite. In the following, such a matrix $\boldsymbol{B}$ will be referred to as a \emph{Bernstein matrix} with supporting points $\boldsymbol{s}_1,\ldots,\boldsymbol{s}_p$.

Other well-known examples of conditionally negative semidefinite matrices are:
\begin{itemize}
    \item the zero ($\boldsymbol{0}$) and all-ones ($\boldsymbol{1}$) matrices; \item the sum of two conditionally negative semidefinite matrices;
    \item the product of a conditionally negative semidefinite matrix with a nonnegative constant;
    \item the limit of a convergent sequence of conditionally negative semidefinite matrices;
    \item the product of two Bernstein matrices having the same supporting points.
\end{itemize}
The second to fourth results stem from the fact that the set of conditionally negative semidefinite matrices is a convex cone and is closed under pointwise convergence. The last result is instead an exception, as, in general, the product of two Bernstein functions is not necessarily a Bernstein function; it stems from \eqref{vario} with $\theta=2$ and the fact that, if $t \mapsto B_1(t)$ and $t \mapsto B_2(t)$ are Bernstein functions, so is $t \mapsto B_1(\sqrt{t})\,B_2(\sqrt{t})$ \citep[corollary 3.8]{porcu2011schoenberg, Schilling}.

\section{Sufficient Validity Conditions for the Multivariate Mat\'ern Model} \label{sec3}

In this section, we establish sufficient validity conditions for the multivariate Mat\'ern model \eqref{multi_cov}, where $\boldsymbol{\alpha}$, $\boldsymbol{\nu}$ and $\boldsymbol{\sigma}$ are symmetric real-valued matrices, the former two with positive entries. We recall that, in what follows, all matrix operations are taken element-wise and we denote by $\left\lfloor \cdot \right\rfloor$ and $\left\lceil \cdot \right\rceil$ the floor and ceil functions, respectively. In the first theorem hereinafter, the direct and cross-covariances are assumed to have the same smoothness parameter $\nu>0$, a restriction that is lifted in the following two theorems.

\begin{theorem}[parsimonious Mat\'ern model]
\label{validityexponentialdD}
The multivariate Mat\'ern model \eqref{multi_cov} with $\boldsymbol{\nu} = \nu \boldsymbol{1}$ and $\nu>0$ is valid in $\mathbb{R}^{d}$, $d \geq 1$, if
\begin{itemize}
    \item [(1)]
    $\boldsymbol{\alpha}$ is conditionally negative semidefinite;
    \item [(2)]
    $\boldsymbol{\sigma} \, \boldsymbol{\alpha}^{\left\lfloor \frac{d+1+3 \lceil 2\nu \rceil}{2} \right\rfloor}$ is positive semidefinite.
\end{itemize}
\end{theorem}

\begin{theorem}[general conditions, full Mat\'ern model]
\label{validitymatern20}
The multivariate Mat\'ern model \eqref{multi_cov} is valid in $\mathbb{R}^{d}$, $d \geq 1$, if, either
\begin{itemize}
\item[(A)]
\begin{itemize}
    \item [(1)] $\boldsymbol{\nu}$
    is a Bernstein matrix;
    \item [(2)]
    $\boldsymbol{\alpha}^{-2}$ is a Bernstein matrix with the same supporting points as $\boldsymbol{\nu}$;
    \item [(3)]
    $\frac{\boldsymbol{\sigma} \, \Gamma(\boldsymbol{\nu}+d/2)}{\boldsymbol{\alpha}^d \, \Gamma(\boldsymbol{\nu})}$ is positive semidefinite;
\end{itemize}
\hspace{-0.7cm} or
\item[(B)]
there exists a symmetric matrix $\boldsymbol{\psi}$
with positive entries such that:
\begin{itemize}
    \item [(1)] $\boldsymbol{\psi}$ is conditionally negative semidefinite;
    \item [(2)] $\boldsymbol{\nu}$ is conditionally negative semidefinite;
    \item [(3)] $\boldsymbol{\alpha}^2 \, \boldsymbol{\psi}-\boldsymbol{\nu}$
    is conditionally negative semidefinite;
    \item [(4)]
    $\frac{\boldsymbol{\sigma}}{\Gamma(\boldsymbol{\nu})} \boldsymbol{\psi}^{\boldsymbol{\nu}+d/2} \, \boldsymbol{\alpha}^{2\boldsymbol{\nu}} \exp(-\boldsymbol{\nu})$ is positive semidefinite.
\end{itemize}
\end{itemize}
\end{theorem}

Conditions (B) above provide general sufficient validity conditions that depend on a conditionally negative semidefinite matrix $\boldsymbol{\psi}$.
Some particular cases are given in the next theorem and examples, which bring simpler sufficient conditions that only depend on the model parameters ($\boldsymbol{\alpha},\boldsymbol{\nu},\boldsymbol{\sigma}$) and, possibly, one scalar hyperparameter $\beta$ instead of a matrix-valued hyperparameter $\boldsymbol{\psi}$.

\begin{theorem}[simplified conditions, full Mat\'ern model]
\label{validityMatern}
The multivariate Mat\'ern model \eqref{multi_cov} is valid in $\mathbb{R}^d$, $d \geq 1$, if, either
\begin{itemize}
\item[(A)]
\begin{itemize}
    \item [(1)] $\boldsymbol{\nu}$ is conditionally negative semidefinite;
    \item [(2)] $\boldsymbol{\nu} \, \boldsymbol{\alpha}^{-2}$ is conditionally negative semidefinite;
    \item [(3)]
    $\frac{\boldsymbol{\sigma}}{\Gamma(\boldsymbol{\nu})} \boldsymbol{\alpha}^{-d} \, \boldsymbol{\nu}^{\boldsymbol{\nu}+d/2} \, \exp(-\boldsymbol{\nu})$ is positive semidefinite;
\end{itemize}
\hspace{-0.7cm} or
\item[(B)] for some positive scalar, $\beta$, one has
\begin{itemize}
    \item [(1)] $\boldsymbol{\nu}$ is conditionally negative semidefinite;
    \item [(2)] $\boldsymbol{\alpha}^2-\beta \boldsymbol{\nu}$ is conditionally negative semidefinite;
    \item [(3)]
    $\frac{\boldsymbol{\sigma}}{\Gamma(\boldsymbol{\nu})}  \left(\frac{\boldsymbol{\alpha}^{2}}{\beta}\right)^{\boldsymbol{\nu}} \exp(-\boldsymbol{\nu})$ is positive semidefinite.
\end{itemize}
\end{itemize}
\end{theorem}

\begin{description}
\item[\textbf{Example 1.}] If $\boldsymbol{\nu}=\nu \boldsymbol{1}$ with $\nu>0$, the following conditions are obtained from Conditions (A) of Theorem~\ref{validityMatern}:
\begin{itemize}
    \item [(1)] $\boldsymbol{\alpha}^{-2}$ is conditionally negative semidefinite;
    \item [(2)]
    $\boldsymbol{\sigma} \boldsymbol{\alpha}^{-d}$ is positive semidefinite.
\end{itemize}
These conditions apply, in particular, to the multivariate exponential model, which corresponds to the case when $\nu$ is equal to $0.5$. Owing to \eqref{convergencetogaussian}, they also apply to the multivariate Gaussian model \eqref{multi_gauss} after replacing $\boldsymbol{\alpha}$ with ${\boldsymbol{\beta}^{1/2}}$, which yields the same conditions as in Lemma~\ref{validityGaussian} given in Appendix~\ref{app1}.\\
\item[\textbf{Example 2.}] If $\boldsymbol{\alpha}=\alpha \boldsymbol{1}$ with $\alpha>0$, the following conditions are obtained from Conditions (A) of Theorem~\ref{validityMatern}:
\begin{itemize}
    \item [(1)] $\boldsymbol{\nu}$ is conditionally negative semidefinite;
    \item [(2)]
    $\frac{\boldsymbol{\sigma}}{\Gamma(\boldsymbol{\nu})} \, \boldsymbol{\nu}^{\boldsymbol{\nu}+d/2} \, \exp(-\boldsymbol{\nu})$ is positive semidefinite.\\
\end{itemize}
\item[\textbf{Example 3.}] If $\boldsymbol{\nu}=\nu \boldsymbol{1}$ with $\nu>0$, the following conditions are obtained from Conditions (B) of Theorem~\ref{validityMatern}:
\begin{itemize}
    \item [(1)] $\boldsymbol{\alpha}^2$ is conditionally negative semidefinite;
    \item [(2)]
    $\boldsymbol{\sigma} \boldsymbol{\alpha}^{2\nu}$ is positive semidefinite.
\end{itemize}
As for Example 1, these conditions apply, in particular, to the multivariate exponential model. Seemingly, the overlap between Examples 1 and 3 is limited to the case when $\boldsymbol{\alpha}=\alpha \, \boldsymbol{1}$, which fulfills the first condition of both examples.\\
\end{description}

\section{Comparison with the Sufficient Conditions in Apanasovich et al. (2012)} \label{sec5}

\subsection{General Statements}

Conditions (B) in Theorem~\ref{validityMatern} bear resemblance to the sufficient conditions provided by \citet{Apanasovich},
specifically:
\begin{itemize}
    \item [(i)] There exist $\delta \geq 0$ and a correlation matrix with nonnegative entries $[a_{ij}]_{i,j=1}^p$ such that
    \begin{equation}
        \label{apanasovich_i}
        \nu_{ij} = \frac{\nu_{ii}+\nu_{jj}}{2}+\delta(1-a_{ij}), \quad i,j=1,\cdots,p;
    \end{equation}
    \item [(ii)] $[\alpha_{ij}^2]_{i,j=1}^p$ is conditionally negative semidefinite;
    \item [(iii)]
    $\left[\frac{\sigma_{ij} \, \Gamma(\nu_{ij}+d/2)}{\Gamma(\nu_{ij}) \, \Gamma((\nu_{ii}+\nu_{jj}+d)/2)}  \alpha_{ij}^{2\delta+\nu_{ii}+\nu_{jj}} \right]_{i,j=1}^p$ is positive semidefinite.
\end{itemize}

In particular, any matrix $\boldsymbol{\nu}$ satisfying Condition (i) of \citet{Apanasovich} also satisfies Condition (B.1) of Theorem~\ref{validityMatern}, but the reciprocal is not true. Indeed, Condition (i) corresponds to a particular case of \eqref{ejemplo}, with $\gamma$ a bounded variogram that does not exceed its sill $\delta$. Variograms with a hole effect that take values greater than the sill, or variograms without a sill, provide examples of matrices $\boldsymbol{\nu}$ satisfying Condition (B.1) of Theorem~\ref{validityMatern} but not Condition (i) of \citet{Apanasovich}; this is the case for the matrix with entries $\nu_{ij}=1+(i-j)^2$ \citep[remark 2.5]{Berg}.
Conversely, any matrix $\boldsymbol{\alpha}$ such that Conditions (B.1) and (B.2) of Theorem~\ref{validityMatern} hold also satisfies Condition (ii) of \citet{Apanasovich}, but the reciprocal is not true. As for Condition (B.3), it does not depend on the spatial dimension $d$, as Condition (iii) of \citet{Apanasovich} does.

In the next subsections, the conditions of Theorems~\ref{validitymatern20} and~\ref{validityMatern} and of \citet{Apanasovich} are compared for three specific examples of multivariate Mat\'ern covariance models.

\subsection{First Example}

Let $2\nu_{ij}=\nu_{ii}+\nu_{jj}$ for $i, j = 1, \cdots, p$. In this case, \eqref{apanasovich_i} gives $\delta=0$, so that Condition (iii) of \citet{Apanasovich} reduces to the positive semidefiniteness of $\boldsymbol{\sigma} \, \boldsymbol{\alpha}^{2\boldsymbol{\nu}} / \Gamma(\boldsymbol{\nu})$, as does Condition (B.3) of Theorem~\ref{validityMatern}, if one accounts for the fact that $\exp(-\nu_{ij})/\beta^{\nu_{ij}}$ separates into the product of a term depending only on $i$ and a term depending only on $j$. Furthermore, if $\boldsymbol{\alpha}^2$ is conditionally negative semidefinite, so is $\boldsymbol{\alpha}^2-\beta \boldsymbol{\nu}$ for any $\beta > 0$. Accordingly, in the case when $2\nu_{ij}=\nu_{ii}+\nu_{jj}$, the conditions of \citet{Apanasovich} coincide with the set (B) of Theorem~\ref{validityMatern} and no longer depend on the spatial dimension $d$.

\subsection{Second Example}

In the two-dimensional space ($d=2)$, consider a $p$-variate Mat\'ern covariance model with $\sigma_{ij}=1$ if $i=j$, $\rho$ otherwise, $\nu_{ij}=0.5$ if $i=j$, $1.5$ otherwise, and $\alpha^2_{ij}=0.5 \beta$ if $i=j$, $1.5 \beta+a$ otherwise, with $\rho \geq 0$, $\beta > 0$ and $a \geq 0$. Conditions (i) and (ii) of \citet{Apanasovich} are satisfied with $\delta=1$, and so are Conditions (B.1) and (B.2) of Theorem~\ref{validityMatern}. Accordingly, one can calculate the maximum permissible value for the collocated correlation coefficient $\rho$ under Conditions (iii) and (B.3), respectively, as a function of $\beta$, $a$ and $p$.

For $a=0$, the maximum collocated correlation coefficient does not depend on $\beta$ and $p$. One finds $\rho_{max}=0.064$ with Condition (iii) of \citet{Apanasovich} and $\rho_{max}=0.523$ with Condition (B.3) of Theorem~\ref{validityMatern}. For $a=0.5$ and $a=5$, the maximum collocated correlation coefficient is found not to depend on $p$, with a much ($8$ to $10.7$ times) higher value obtained with Condition (B.3) of Theorem~\ref{validityMatern} than with Condition (iii) of \citet{Apanasovich} (Figure~\ref{fig:caso3}). The difference between both sets of conditions increases even more when the spatial dimension $d$ increases, because Condition (iii) of \citet{Apanasovich} becomes stricter and stricter, while Condition (B.3) of Theorem~\ref{validityMatern} remains unchanged.

\begin{figure}
    \centering
\includegraphics[width = 0.49\textwidth]{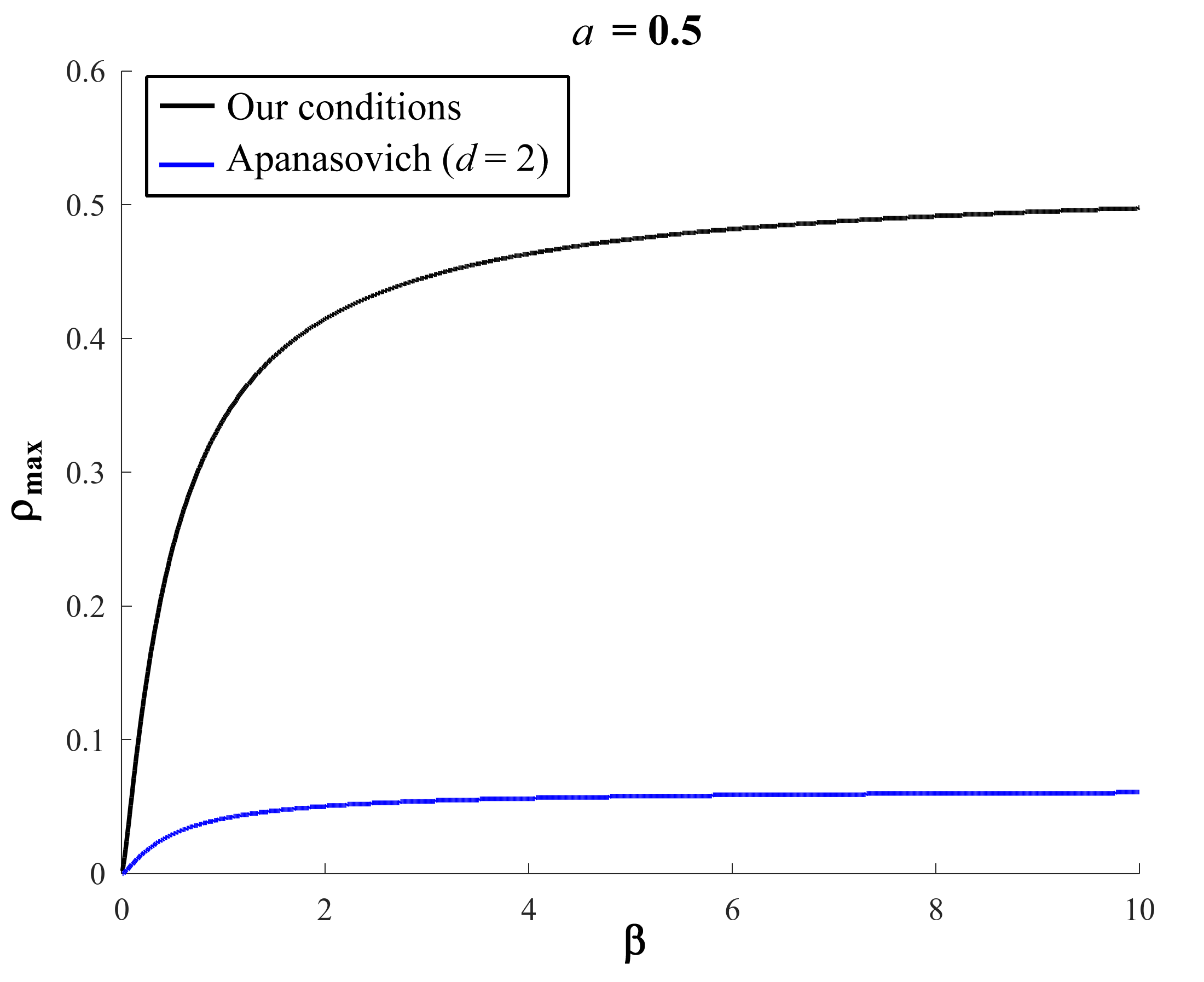}
\includegraphics[width = 0.49\textwidth]{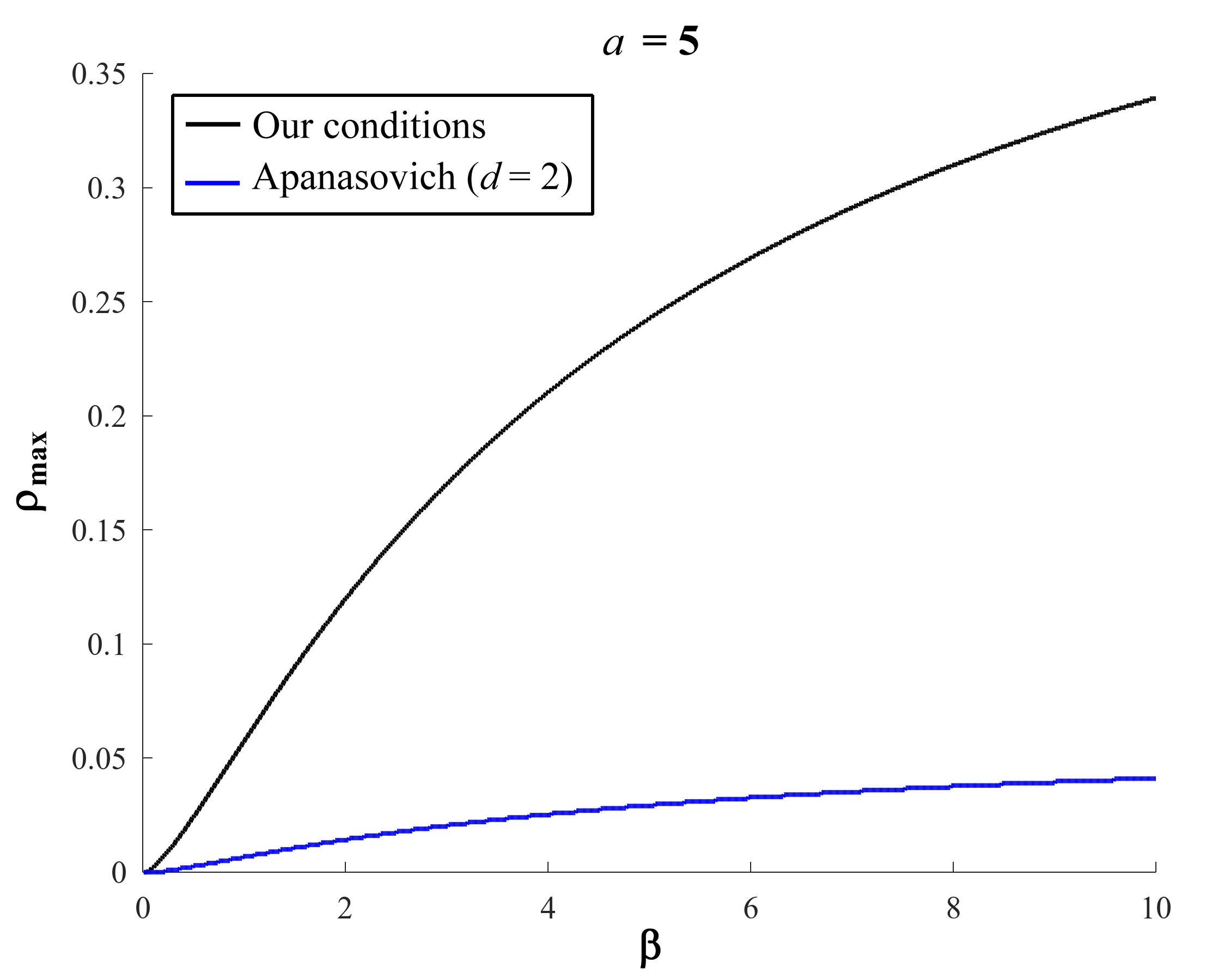}
    \caption{Upper bounds for $\rho$ as a function of $\beta$ for $a=0.5$ (left) and $a=5$ (right), based on Apanasovich's conditions and Conditions (B) of Theorem~\ref{validityMatern}}.
    \label{fig:caso3}
\end{figure}

\subsection{Third Example}

We consider the same spatial dimension ($d=2$) and parameterization for $\boldsymbol{\nu}$ and $\boldsymbol{\sigma}$ as in the previous example, but modify the scale parameters as follows:
$\alpha_{ij}^2 = a>0$ if $i=j$, $\ln(1+a)$ otherwise. In such a case, Conditions (ii) of \citet{Apanasovich} and (B.2) of Theorem~\ref{validityMatern} are not satisfied, since $\boldsymbol{\alpha}^2$ is not conditionally negative semidefinite ($\alpha_{ij}^2<\min(\alpha_{ii}^2,\alpha_{jj}^2)$ for $i \neq j$).

However, Conditions (A.1) and (A.2) of Theorem~\ref{validitymatern20} and Conditions (A.1) and (A.2) of Theorem~\ref{validityMatern} are satisfied with the chosen parameters, so that both theorems can be used to derive an upper bound for the collocated correlation coefficient $\rho$. Figure~\ref{fig:caso4} compares the maximal value for $\rho$ leading to a positive semidefinite matrix when applying Conditions (A.3) of Theorem~\ref{validitymatern20} and (A.3) of Theorem~\ref{validityMatern}, for $a$ ranging from $0.01$ to $10$. The upper bound found with the former condition is indistinguishable from that derived from the necessary and sufficient conditions provided by \citet[theorem 3]{Gneiting:Kleibler:Schlather:2010} in the bivariate setting. Gneiting's conditions apply only for $p=2$, whereas that of Theorems~\ref{validitymatern20} and~\ref{validityMatern} have been established for any positive integer $p$.

\begin{figure}
    \centering
\includegraphics[width = 0.65\textwidth]{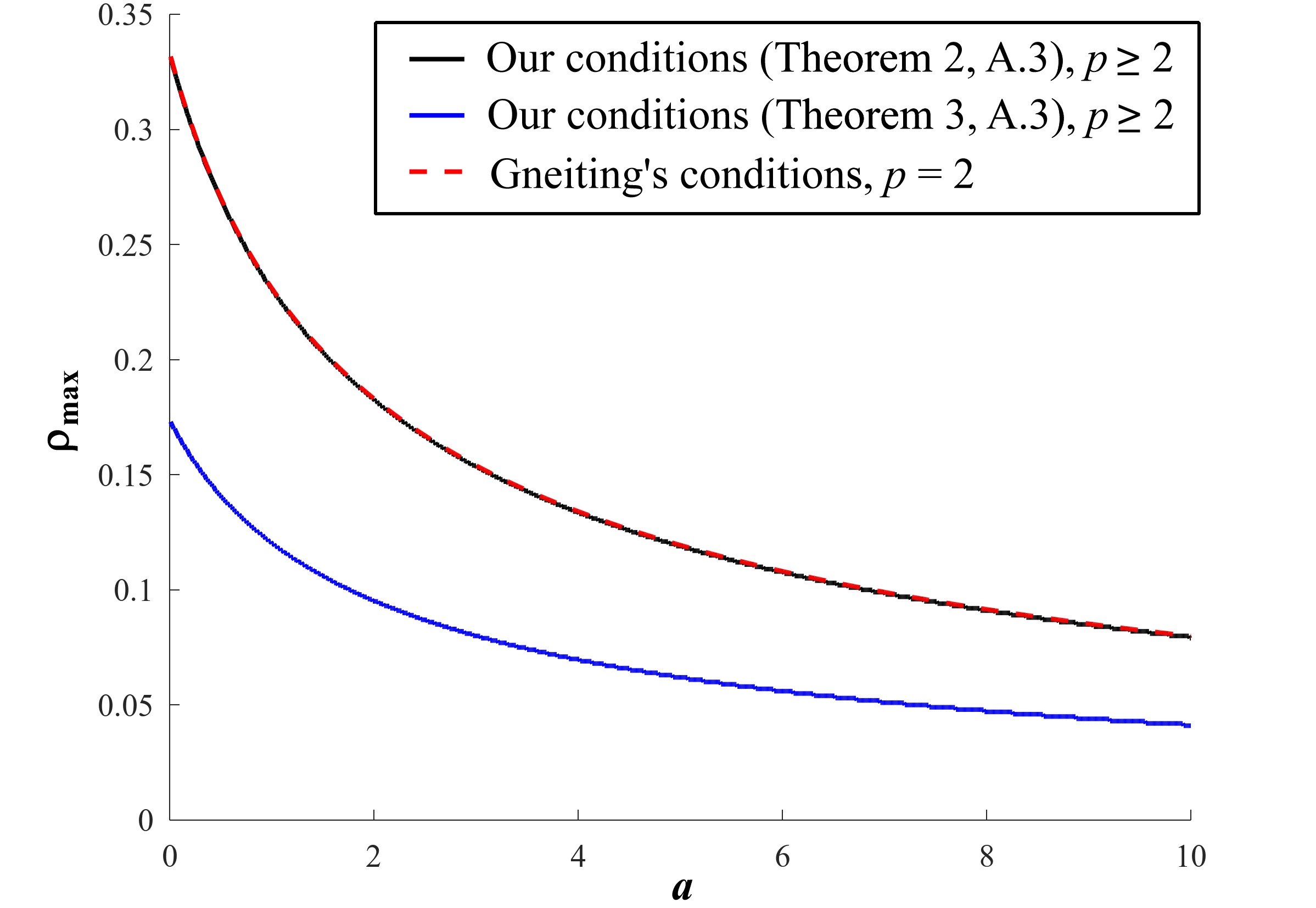}
    \caption{Upper bounds for $\rho$ as a function of $a$, based on Gneiting's conditions (valid for $p=2$ only), Conditions (A) of Theorem~\ref{validitymatern20} and Conditions (A) of Theorem~\ref{validityMatern} (valid for any $p$).}
    \label{fig:caso4}
\end{figure}

\section{Trivariate Geochemical Data Illustration} \label{sec6}

The goal of this section is to show that parameter estimation is feasible for the proposed parameterizations, and that fitting performances are competitive with respect to previous parameterizations proposed in earlier literature. Proposing innovative estimation techniques or comparing different existing techniques is not of interest here.

We consider a data set of geochemical concentrations measured at $n = 1752$ surface samples from a 5750 m by 2400 m area in southern Equator, aimed at identifying targets for mining exploration \citep{guartan2, guartan}. This data set was selected because it has empirical attributes that even challenge our extended validity constraints.
Here, we focus on the concentrations of bismuth (Bi), magnesium (Mg), and tellurium (Te), which are transformed to Gaussianity using an empirical normal scores transformation. We scale each transformed variable so that it has a sample mean of zero and a sample variance of one.
Tellurium and bismuth often co-occur, while magnesium appears at lower levels when tellurium and bismuth concentrations are high (Figure~\ref{fig:concentration}). The empirical binned direct and cross-variograms (Figure~\ref{fig:semi_over}) show significant spatial auto-correlations, cross-covariances, and nugget effects.

\begin{figure}[ht]
\centering
\includegraphics[width = 0.32\textwidth]{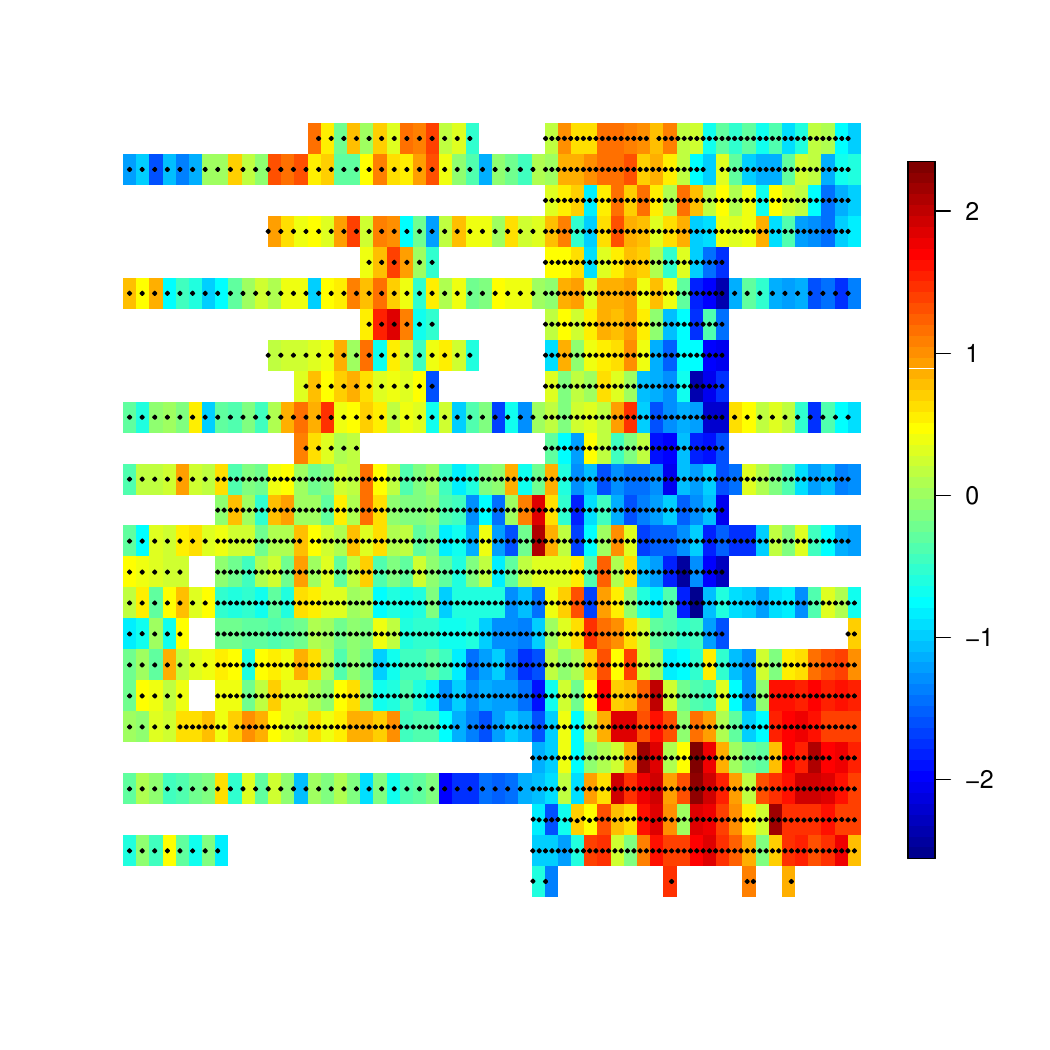}
\includegraphics[width = 0.32\textwidth]{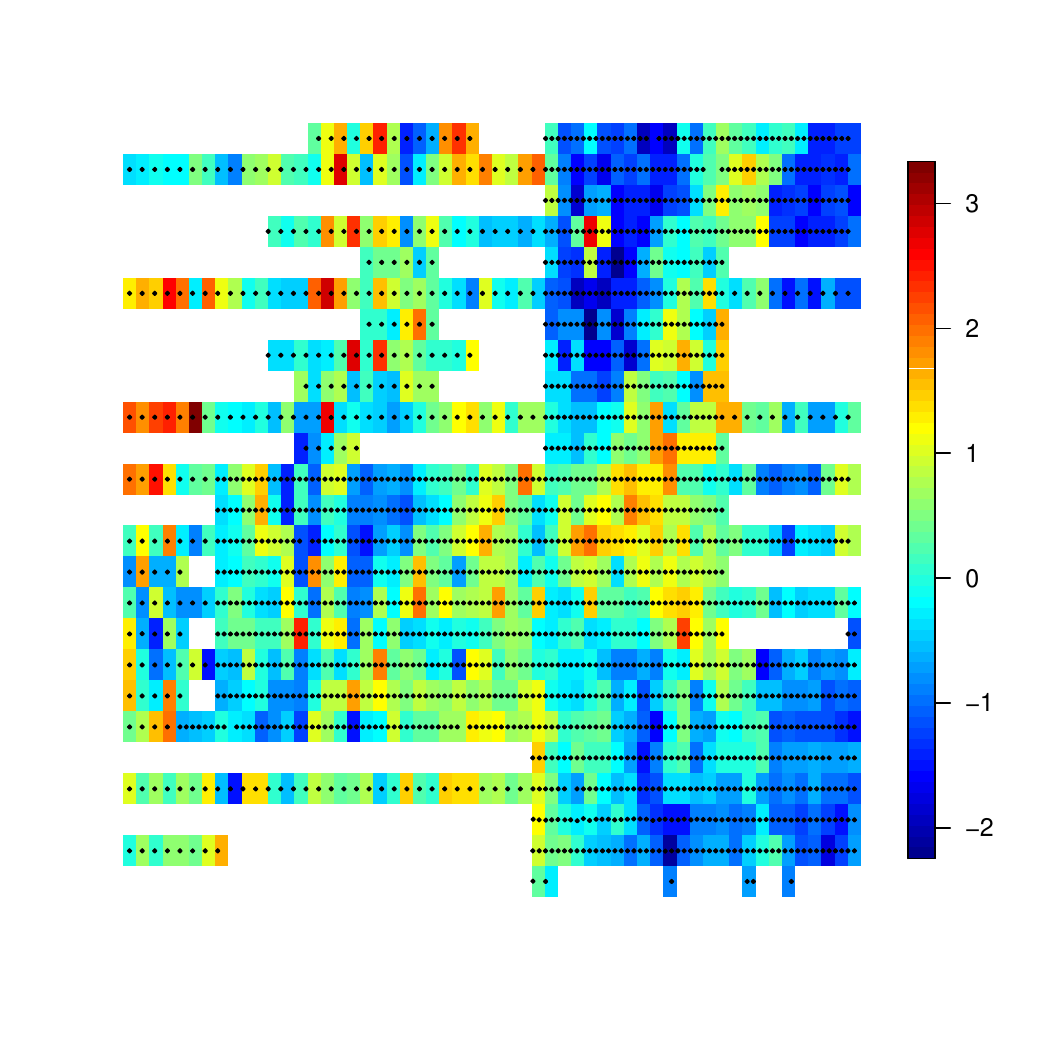}
\includegraphics[width = 0.32\textwidth]{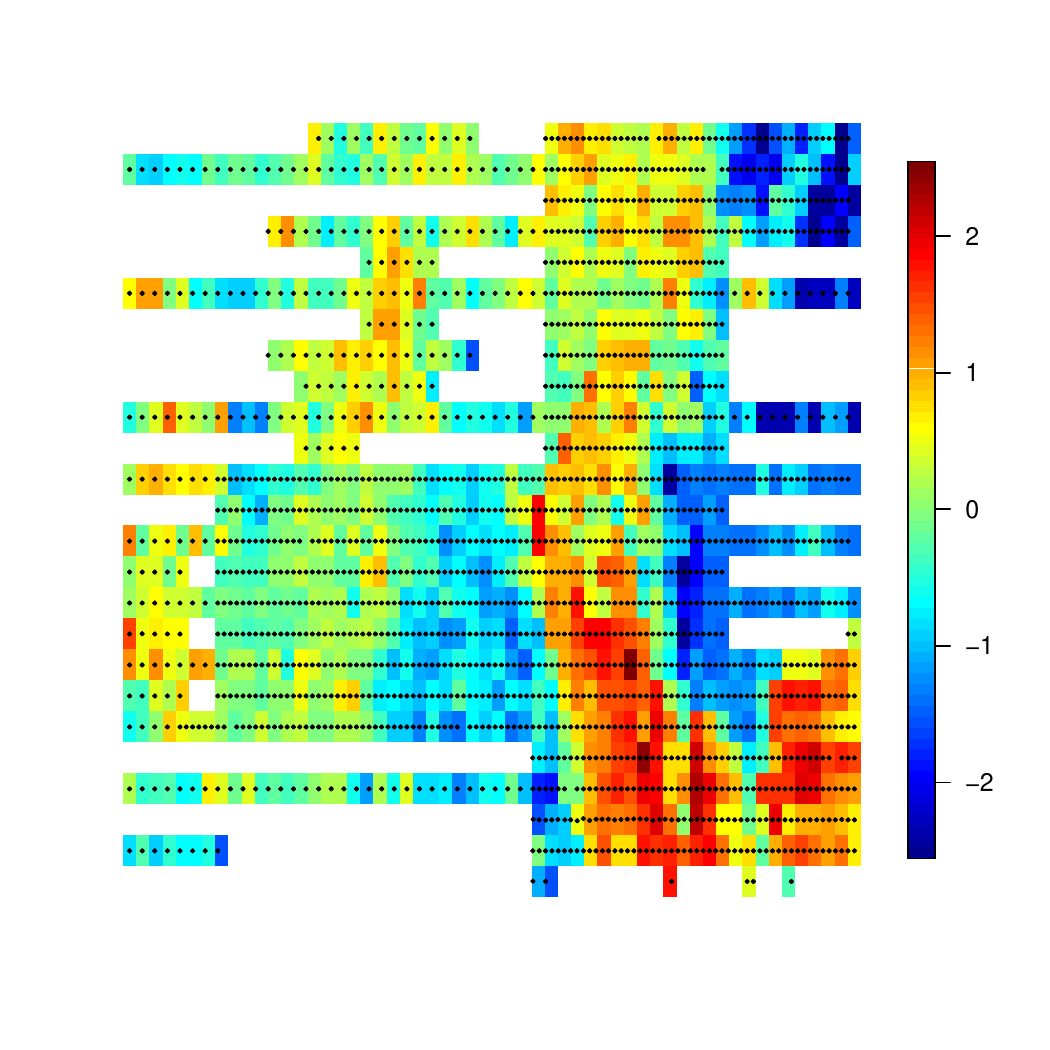}

\caption{Transformed concentrations of bismuth (left), magnesium (center), and tellurium (right).}\label{fig:concentration}
\end{figure}

As an exploratory tool, we fit exponential covariance models to the empirical direct and cross-variograms using weighted least-squares, with weights equal to the ratio of the number of pairs over the squared distance. Let $\hat{\boldsymbol{\alpha}}$ and $\hat{\boldsymbol{\sigma}}$ be the corresponding fitted scale and collocated covariance parameters. This data set is interesting because it matches some, but never all, of the constraints for our proposed classes, \cite{Apanasovich} and \cite{Gneiting:Kleibler:Schlather:2010}. We find that the entries of $\hat{\boldsymbol{\alpha}}$ differ by as much of a factor of two. We also find that $\hat{\boldsymbol{\alpha}}^{-2}$ is conditionally negative semidefinite, $\hat{\boldsymbol{\alpha}}^{2}$ is not conditionally negative semidefinite, $\hat{\boldsymbol{\sigma}}$ is positive semidefinite, $\hat{\boldsymbol{\sigma}} \hat{\boldsymbol{\alpha}}$ is positive semidefinite, and $\hat{\boldsymbol{\sigma}} \hat{\boldsymbol{\alpha}}^{-2}$ is not positive semidefinite.

We compare these attributes of the empirical fits to the constraints imposed by \cite{Gneiting:Kleibler:Schlather:2010}, \cite{Apanasovich}, and Example 1 of Theorem~\ref{validityMatern}A, under the assumption that $\boldsymbol{\nu}  =  \mathbf{1}/2$ outlined in Table 1. All notations and matrix operations are those defined in Table 1.

\begin{description}

\item[\textbf{\cite{Gneiting:Kleibler:Schlather:2010}}:] $\boldsymbol{\alpha} = \alpha \mathbf{1}$ and $\boldsymbol{\sigma} \geq 0$. Given the variability in $\hat{\boldsymbol{\alpha}}$, the constant scale parameter $\boldsymbol{\alpha} = \alpha \mathbf{1}$ may be too restrictive. Within the constraints of \cite{Gneiting:Kleibler:Schlather:2010}, $\hat{\boldsymbol{\sigma}}$ is positive semidefinite.

\item[\textbf{\cite{Apanasovich}}:] $\boldsymbol{\alpha}^{2} \leq 0$ and $\boldsymbol{\sigma} \boldsymbol{\alpha} \geq 0$. This corresponds to $\delta = 0$ and $a_{ij} = 1$ in \cite{Apanasovich}. Thus, the scale parameter constraint (ii) of \cite{Apanasovich} is not appropriate. However, the element-wise product of $\hat{\boldsymbol{\sigma}}$ and $\hat{\boldsymbol{\alpha}}$ is positive semidefinite, matching constraint (iii) of \cite{Apanasovich} when $\boldsymbol{\nu}  = \nu \mathbf{1}$, where $\nu = 1/2$.

\item[\textbf{Theorem~\ref{validityMatern}A, Example 1:}] $\boldsymbol{\alpha}^{-2} \leq 0$ and $\boldsymbol{\sigma} \boldsymbol{\alpha}^{-2} \geq 0$. $\hat{\boldsymbol{\alpha}}^{-2} \leq 0$, suggesting that our constraints may be preferable to those of \cite{Apanasovich} and \cite{Gneiting:Kleibler:Schlather:2010}. On the other hand, the element-wise product of $\hat{\boldsymbol{\sigma}}$ and $\hat{\boldsymbol{\alpha}}^{-2}$ is not positive semidefinite, meaning that our constraints do not match the empirical direct and cross-variogram fits perfectly.

\end{description}
Although each set of constraints fails to capture all of the empirical attributes of the data, our newly proposed constraints offer practical advantages over the constraints of \cite{Apanasovich} and \cite{Gneiting:Kleibler:Schlather:2010} for this data set, as is now shown.
We assume that the transformed metal concentrations $\bm{Z}(\boldsymbol{s}_i)$ at location $\boldsymbol{s}_i \in \mathbb{R}^2$, $i = 1,...,n$, are normally distributed with mean zero and
\begin{equation}
\text{Cov}\left(\bm{Z}( \boldsymbol{s}_i),\bm{Z}(\boldsymbol{s}_j)\right) = \textbf{K}\left(\boldsymbol{s}_i - \boldsymbol{s}_j; \boldsymbol{\alpha}, \frac{\mathbf{1}}{2},\boldsymbol{\sigma} \right) + \textbf{V} \, I\left(i = j\right ),
\end{equation}
where $\textbf{K}(\boldsymbol{s}_i - \boldsymbol{s}_j; \boldsymbol{\alpha}, \mathbf{1}/2,\boldsymbol{\sigma})$ is defined in \eqref{multi_cov}, $\textbf{V}$ is a positive semidefinite matrix, and $I(\cdot)$ is the indicator function representing a nugget effect covariance.

For each set of constraints, we fit a coregionalization model in a Bayesian framework assuming \textit{a priori} that $\boldsymbol{\alpha}$ and $\boldsymbol{\sigma}$ are uniformly distributed under the constraints discussed above. Therefore, although the models are the same, the parameter constraints imposed are not. Because of identifiability challenges for $\boldsymbol{\sigma}$ and $\boldsymbol{\alpha}$ \citep{zhang2004}, we also consider models with $\boldsymbol{\alpha}$ fixed on the configuration closest to $\hat{\boldsymbol{\alpha}}$, given the appropriate constraints. We have six models using the constraints discussed in Theorem~\ref{validityMatern}A Example 1, \cite{Gneiting:Kleibler:Schlather:2010}, and \cite{Apanasovich}, as well as letting $\boldsymbol{\alpha}$ be fixed or random. We emphasize that these models are the same except for the constraints on the covariance parameters.

To provide a familiar baseline comparison, we also fit a linear model of coregionalization (LMC). To specify the covariance function for this model, we use two spherical covariance functions with fixed ranges of 200 m and 1000 m, as well as a nugget effect. The nugget effect and these spatial ranges appear reasonable selections given the empirical variograms in Figure~\ref{fig:semi_over}. For each of these three terms, there are six parameters (partial sills), giving 18 in total. We assume flat prior distributions for these parameters. Thus, this LMC model has similar complexity compared to the multivariate Mat\'ern models against which it is compared.

We fit all the models using adaptive Markov chain Monte Carlo (MCMC) \citep{roberts2009}. Because the parameter values obtained by fitting the empirical variograms do not satisfy the multivariate Mat\'ern constraints, we initialize the MCMC at the closest parameter configuration that satisfies the appropriate constraints. We use a multivariate normal random walk for all parameters using the log-scale for positive parameters ($\boldsymbol{\alpha}$, as well as diagonals of $\boldsymbol{\sigma}$ and $\textbf{V}$) and the raw scale for other parameters (off-diagonals of $\boldsymbol{\sigma}$ and $\textbf{V}$). New parameter proposals are accepted with the probability given by the Metropolis algorithm. Because we impose the specific constraints required by \cite{Gneiting:Kleibler:Schlather:2010}, \cite{Apanasovich}, and Example 1 of Theorem~\ref{validityMatern}A through the prior distribution, any proposal that fails to meet the appropriate constraints is rejected because the proposal has an acceptance probability of 0. To check the conditional negative semidefiniteness of a matrix, we utilize Lemma~\ref{superlemma}.
We tune the covariance of the multivariate normal proposal distribution using the empirical covariance of previous samples and adding a small number to the diagonal of the empirical covariance \citep{haario2001}. We scale the proposal covariance to obtain an acceptance rate between 0.15 and 0.5. Using this approach, we run the MCMC sampler for 60,000 iterations, discarding the first 30,000 iterations.

We compare models using the deviance information criterion (DIC) \citep{spiegelhalter2002bayesian}. In many settings, other information criteria are better approximations for cross-validation than DIC \citep[see, e.g.,][]{watanabe2010asymptotic}; however, calculating the widely applicable information criterion (WAIC) requires arbitrary partitioning of the spatial domain \citep{gelman2014understanding}. In this case, the models are similar in form and complexity. The model comparison results are given in Table~\ref{tab:modcomp}.

First, we note that all of the models using the multivariate Mat\'ern covariance outperformed the LMC model by nearly 100 in terms of DIC. Thus, not only does the multivariate Mat\'ern covariance yield a practical advantage over the LMC, it also has more interpretable parameters. For fixed $\boldsymbol{\alpha}$, our model from the expanded constraints presented in Example 1 of Theorem~\ref{validityMatern}A provides the best model in terms of DIC by about 11. For random or unknown $\boldsymbol{\alpha}$, our model from Example 1 of Theorem~\ref{validityMatern}A is best by about 10 in terms of DIC. For models with fixed $\mathbf{\alpha}$ (1-3 in Table~\ref{tab:modcomp}), the models have essentially the same size not accounting for the imposed constraints. Therefore, differences greater than 10 in the DIC and posterior mean deviance $\overline{D}$ are striking as these differences correspond to added model flexibility given by Theorem~\ref{validityMatern}. For models with unknown or random $\mathbf{\alpha}$ (4-6 in Table~\ref{tab:modcomp}), the model by  \cite{Gneiting:Kleibler:Schlather:2010} has fewer parameters because it has only one range parameter. Despite this simplicity, it has the highest (worst) DIC. Although the model size under the constraints by  \cite{Apanasovich} and Example 1 of Theorem~\ref{validityMatern}A are the same, we observe differences in DIC and posterior mean deviance greater than 10.
In summary, this data set provides an example where the more flexible constraints presented herein offer a better fit to the data than the constraints in \cite{Gneiting:Kleibler:Schlather:2010} or \cite{Apanasovich}.

\begin{table}[ht]
\centering
\scriptsize
\begin{tabular}{rlllll}
  \hline
 & Constraints & $\boldsymbol{\alpha}$ & Mean Deviance $\overline{D}$ & Model Complexity $p_D$ &  $ \text{DIC} = \overline{D} + p_D$ \\
  \hline
1 & Gneiting & Fixed & 9505.65 & 10.34 & 9515.99 \\
2 & Apanasovich & Fixed & 9504.09 & 11.87 & 9515.95 \\
3 & Th. 3A Ex.1 & Fixed & 9494.11 & 10.65 & 9504.75 \\
  \hline
4 & Gneiting & Random & 9486.63 & 12.78 & 9499.42 \\
5 & Apanasovich & Random & 9484.13 & 11.60 & 9495.73  \\
6 & Th. 3A Ex.1 & Random& 9472.56 & 13.15 & 9485.71 \\
  \hline
7 & LMC & ---------- & 9595.04 & 16.90 & 9611.94 \\
   \hline
\end{tabular}
\caption{Model comparison results. Constraints represent the parameter constraints for the three models. We compare models with fixed or unknown $\boldsymbol{\alpha}$. All components of DIC are provided: Mean Deviance $\overline{D}$, Model Complexity $p_D$, and $ \text{DIC} = \overline{D} + p_D$.}\label{tab:modcomp}
\end{table}

For the lowest DIC model, we plot realizations of the variogram for the posterior samples as a function of the lag separation distance in Figure~\ref{fig:semi_over}. We thin the 30,000 post-burn posterior samples to 1000 samples and plot the corresponding direct and cross-variograms in Figure~\ref{fig:semi_over}. Three to four variograms align relatively well with the empirical variograms. We suggest two possible reasons why the variogram fits may not appear ideal. First, limitations on the joint identifiability of scale and collocated covariance parameters are well-studied for the Mat\'ern covariance class \citep{zhang2004}. That is, equal values for $\boldsymbol{\alpha} \boldsymbol{\sigma}$ yield equivalent Gaussian measures and interpolation. However, these non-identifiable combinations of scale and collocated covariance parameters yield different variograms. Thus, if there is a goal of recovering the empirical variogram, one possible remedy includes fixing range parameters subject to values near the least-squares solution. However, as seen in Table~\ref{tab:modcomp}, fixing the range parameters is detrimental to the model performance. Second, this data set was selected to challenge even our improved parametric constraints, as well as those in \cite{Gneiting:Kleibler:Schlather:2010} and \cite{Apanasovich}. As discussed, the element-wise product of $\hat{\boldsymbol{\sigma}}$ and $\hat{\boldsymbol{\alpha}}^{-2}$ is not positive semidefinite; therefore, this constraint from Theorem~\ref{validityMatern}A may limit how well the Mat\'ern model fits the data. That said, our proposed constraints in Example 1 of Theorem~\ref{validityMatern} outperform the current state of the art for the multivariate Mat\'ern model.

\begin{figure}[H]
\includegraphics[width=0.325\textwidth]{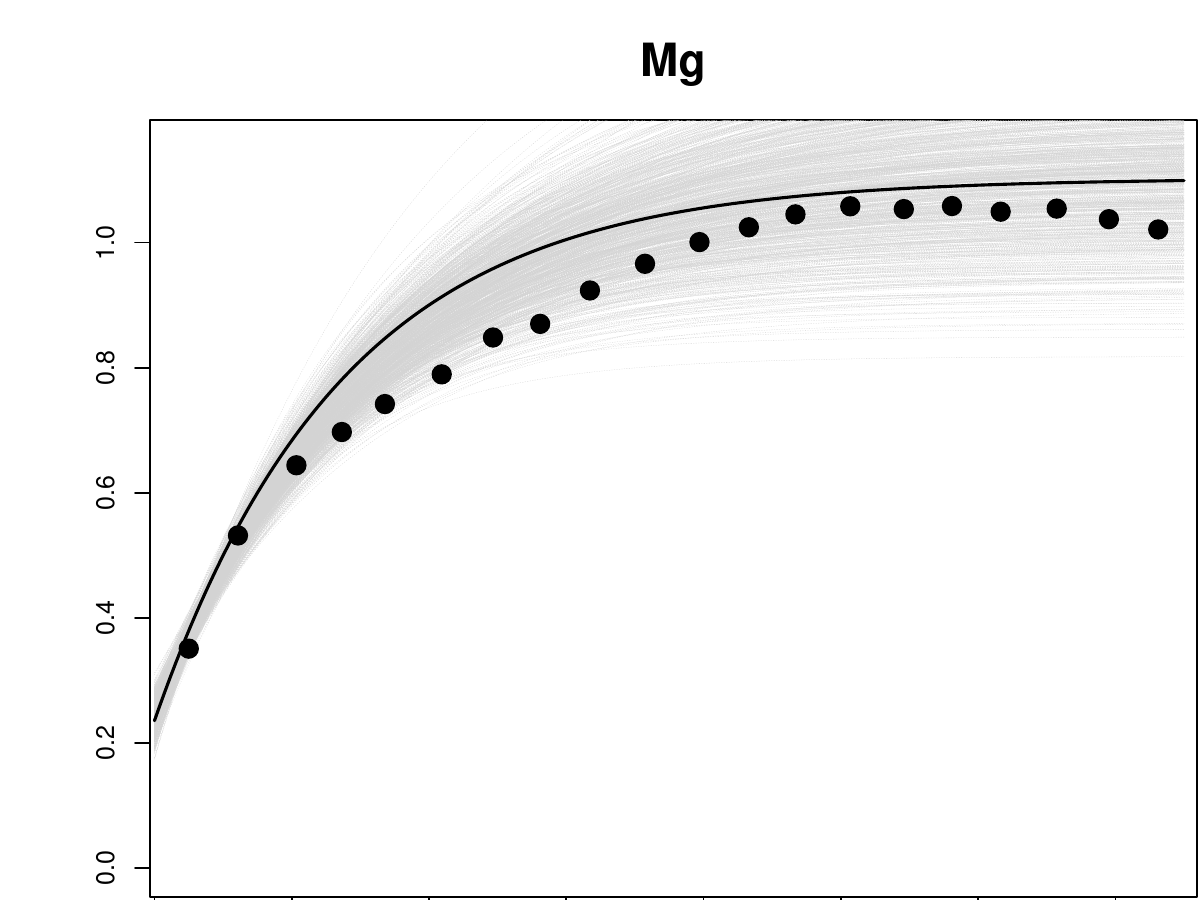} \\
\includegraphics[width=0.325\textwidth]{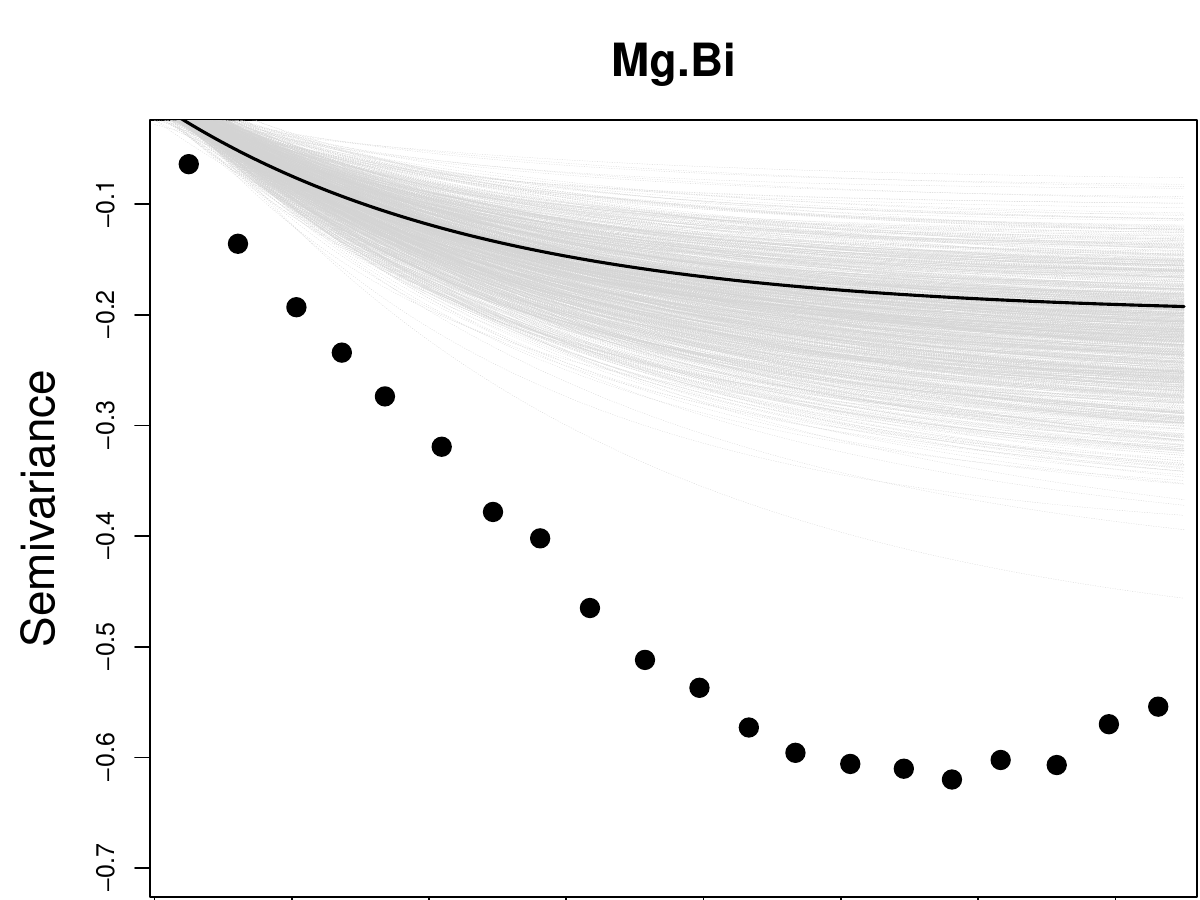}
\includegraphics[width=0.325\textwidth]{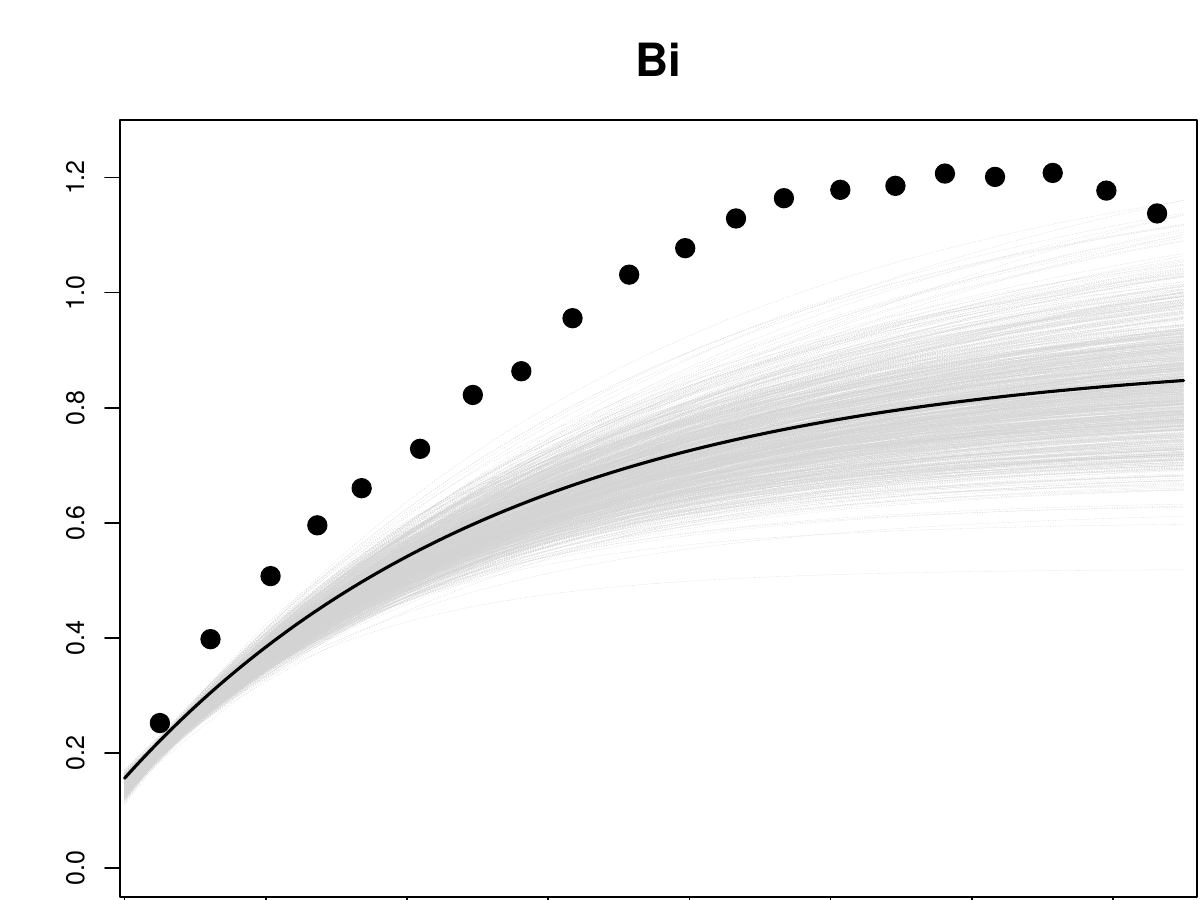} \\
\includegraphics[width=0.325\textwidth]{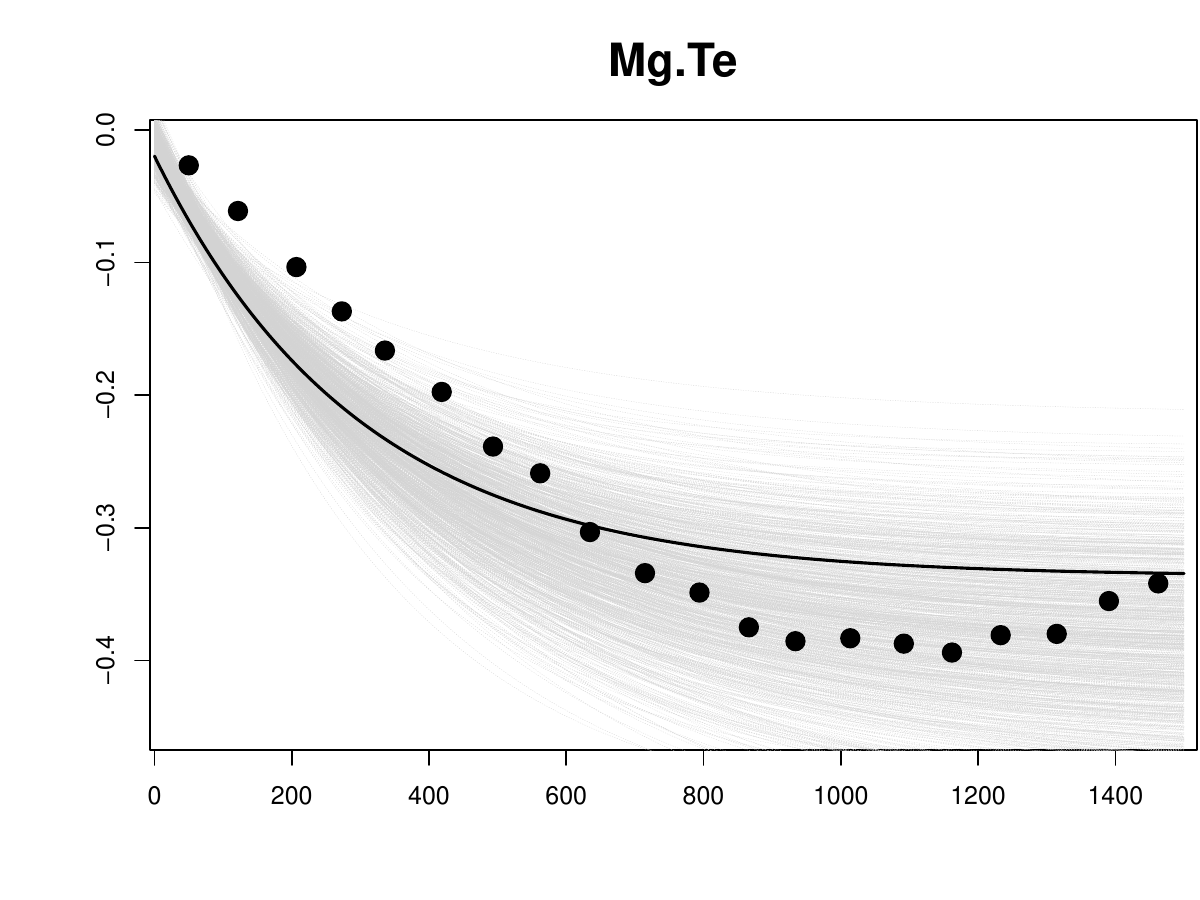}
\includegraphics[width=0.325\textwidth]{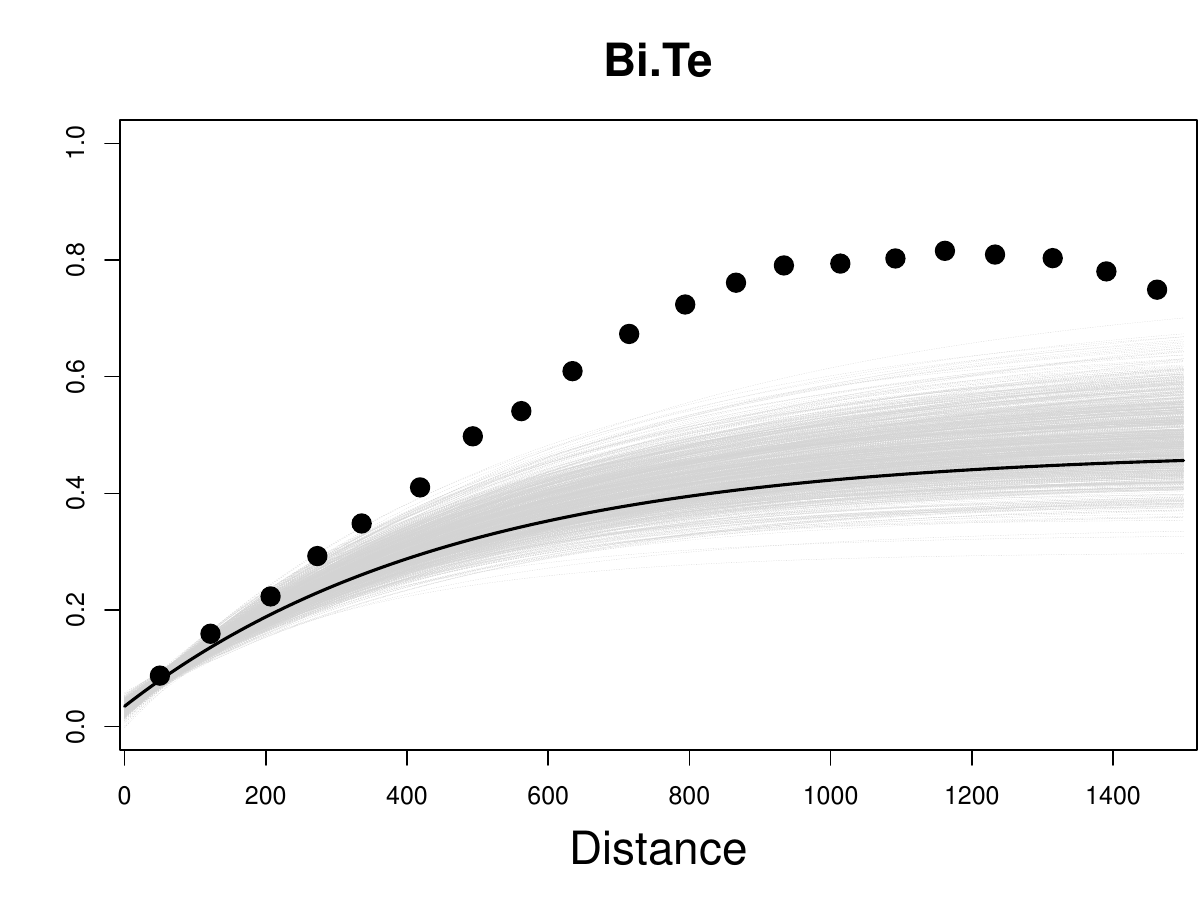}
\includegraphics[width=0.325\textwidth]{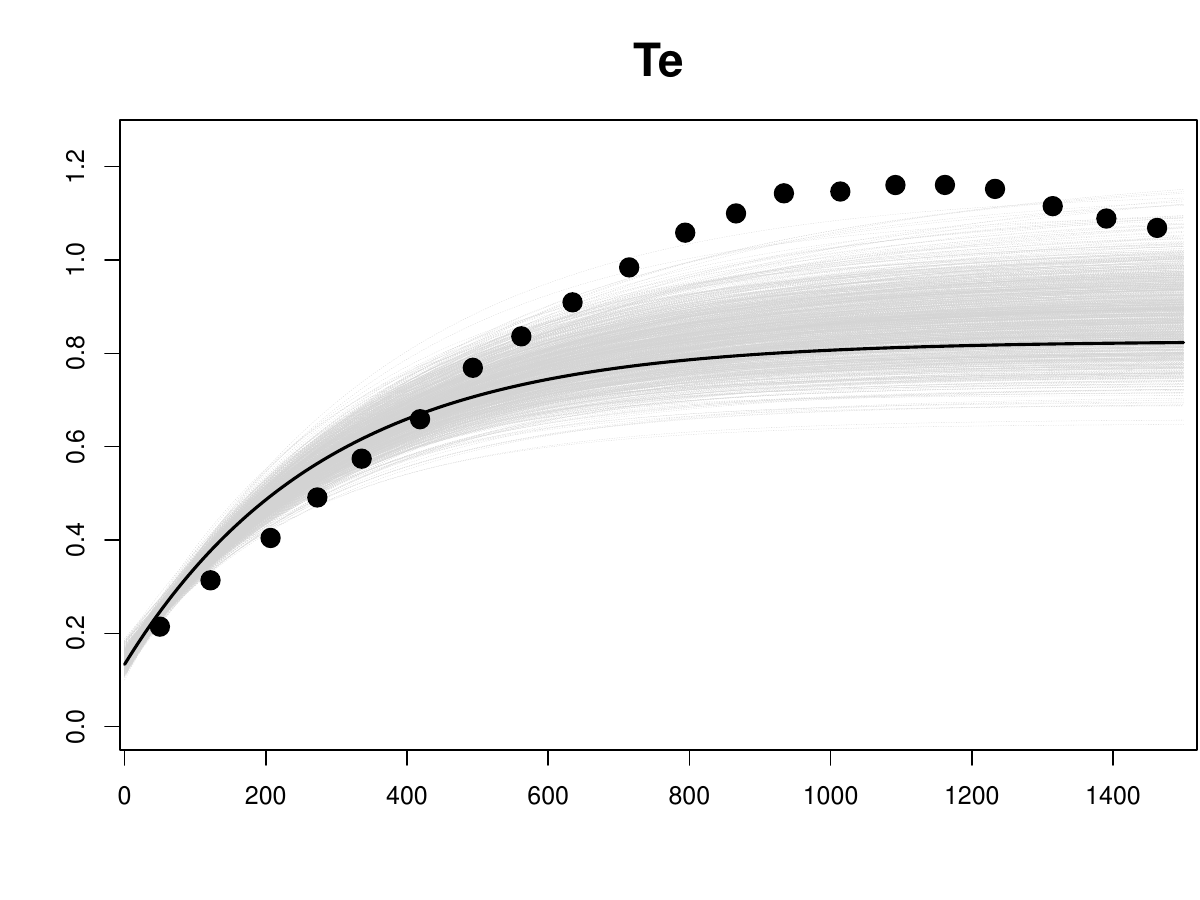}
\caption{Empirical variograms (points) with 1000 posterior variograms superimposed (dashed gray lines). We also include the posterior mean variogram as a solid black line.}\label{fig:semi_over}
\end{figure}

\section{Conclusions} \label{sec7}

While \cite{Gneiting:Kleibler:Schlather:2010} completely characterized the multivariate Mat\'ern model for the case $p=2$, when $p>2$ the sufficient conditions proposed first by \cite{Gneiting:Kleibler:Schlather:2010} and later on by \cite{Du2012} and \cite{Apanasovich} are overly restrictive. This paper has provided improved conditions (Table~\ref{resume}) that allow for a more flexible modeling using the multivariate Mat{\'e}rn covariance and that have been established by use of spectral and scale mixture representations of the Mat\'ern covariance. The loosest conditions are clearly those given in Theorem~\ref{validitymatern20}, of which Theorem~\ref{validityMatern} and the related Examples 1 to 3 are particular cases under simplified parameterizations. Theorem~\ref{validityexponentialdD} provides another set of conditions that ensure the validity of the multivariate Mat\'ern model, which, although restricted to direct and cross-covariances having the same smoothness parameter, can outperform the conditions provided in Theorem~\ref{validitymatern20} for this particular case.

The examples presented in Section~\ref{sec5} give a clear illustration of how the conditions provided by \cite{Apanasovich} can become overly severe, especially in terms of upper bound for the collocated correlation coefficients. In particular, our examples show that one can attain much higher upper bounds for the collocated correlation coefficient, resulting in more flexible versions of the multivariate Mat{\'e}rn model. In some circumstances, the conditions found for the case $p>2$ are as good as in the case $p=2$ provided by \cite{Gneiting:Kleibler:Schlather:2010}. The greater flexibility is also reflected in Section~\ref{sec6} where our parameterization was successful compared to our competitors on a trivariate data set.

Modeling multivariate spatial data sets has an obvious computational burden, in particular when $p>2$, as the number of parameters to be estimated becomes massive. An alternative to least-squares fitting is to proceed through Bayesian routines coupled with MCMC techniques, as we have shown in Section~\ref{sec6}.

\section*{Acknowledgements}

This work was supported by the National Agency for Research and Development of Chile [grants ANID/FONDECYT/REGULAR/No.\ 1210050 and ANID PIA AFB180004].

\appendix



\section{Technical Lemmas with their Proofs}
\label{app1}

From a notation viewpoint, we recall that, in what follows, all the matrix operations (product, division, inverse, power, exponential, integration, etc.) are taken element-wise.

\begin{lem}
\label{superlemma}
Let $\boldsymbol{A}=[a_{ij}]_{i,j=1}^p$ be a symmetric real-valued matrix. The following assertions are equivalent \citep{Schoenberg1938, Berg, Reams}:
\begin{itemize}
    \item $\boldsymbol{A}$ is conditionally negative semidefinite;
    \item $\exp(- t \, \boldsymbol{A})$ is positive semidefinite for all $t \geq 0$;
    \item $[a_{ip}+a_{pj}-a_{ij}-a_{pp}]_{i,j=1}^p$ is positive semidefinite. \\
\end{itemize}
\end{lem}

\begin{lem}
\label{lemma2}
If $\boldsymbol{A}=[a_{ij}]_{i,j=1}^p$ is a conditionally negative semidefinite matrix with positive entries, then $[a^{-1}_{ij}]_{i,j=1}^p$ is positive semidefinite and $[a^{\mu}_{ij}]_{i,j=1}^p$ is conditionally negative semidefinite for any $\mu \in (0,1]$.
\end{lem}
The proof of this lemma can be found in \citet[corollary 2.10 and exercise 2.21]{Berg}.\\

\begin{lem}
\label{lemma3}
The multivariate Mat\'ern model \eqref{multi_cov}
is valid in $\mathbb{R}^d$, $d \geq 1$, if
\begin{itemize}
    \item [(1)] $\boldsymbol{\nu}$ is conditionally negative semidefinite;
    \item [(2)] ${\bf K}({\boldsymbol{h}};\boldsymbol{\alpha},\boldsymbol{\nu}+\mu,\Gamma(\boldsymbol{\nu}+\mu) \, \Gamma(\boldsymbol{\nu})^{-1} \, \boldsymbol{\sigma})$ is a valid model in $\mathbb{R}^d$ for some positive constant $\mu$, where the exponent $-1$ stands for the element-wise inverse and $\Gamma$ for the gamma function.
\end{itemize}
\end{lem}

\begin{proof}
The proof relies on the decomposition of the Mat\'ern covariance with smoothness parameter $\boldsymbol{\nu}$ as a scale mixture of Mat\'ern covariances with smoothness parameter $\boldsymbol{\nu}+\mu$ \citep[formula 6.592.12]{Grad}:
\begin{equation}
    {\bf K}(\boldsymbol{h};\boldsymbol{\alpha},\boldsymbol{\nu},\boldsymbol{\sigma}) = \Gamma(\boldsymbol{\nu}+\mu) \, \Gamma(\boldsymbol{\nu})^{-1} \, \int_{1}^{+\infty} t^{-\boldsymbol{\nu}-\mu} \, (t-1)^{\mu-1} {\bf K}(\boldsymbol{h}\sqrt{t};\boldsymbol{\alpha},\boldsymbol{\nu}+\mu,\boldsymbol{\sigma}) \text{d}t,
\end{equation}
where all operations are taken element-wise. Under Conditions (1) and (2) of Lemma~\ref{lemma3}, ${\bf K}(\boldsymbol{h};\boldsymbol{\alpha},\boldsymbol{\nu},\boldsymbol{\sigma})$ appears as a sum of valid multivariate Mat\'ern covariance functions of the form  ${\bf K}(\boldsymbol{h}\sqrt{t};\boldsymbol{\alpha},\boldsymbol{\nu}+\mu,\Gamma(\boldsymbol{\nu}+\mu) \, \Gamma(\boldsymbol{\nu})^{-1} \, \boldsymbol{\sigma})$ weighted by positive semidefinite matrices of the form $t^{-\boldsymbol{\nu}-\mu} \, (t-1)^{\mu-1}$ (Lemma~\ref{superlemma}), hence, it is a valid covariance model.
\end{proof}

\begin{cor}
\label{corollary3}
If ${\bf K}({\boldsymbol{h}};\boldsymbol{\alpha},\nu \boldsymbol{1},\boldsymbol{\sigma})$ is a valid model in $\mathbb{R}^d$ for some $\nu>0$, then so is ${\bf K}({\boldsymbol{h}};\boldsymbol{\alpha},\mu \boldsymbol{1},\boldsymbol{\sigma})$ for $\mu \in (0,\nu)$.
\end{cor}
\begin{proof}
This result stems from the fact that the constant matrix $\nu \boldsymbol{1}$ is conditionally negative semidefinite, hence it satisfies Condition (1) of Lemma~\ref{lemma3}.
\end{proof}

\begin{lem}
\label{validityGaussian}
The multivariate Gaussian model \eqref{multi_gauss} is valid in $\mathbb{R}^d$, $d \geq 1$, if \begin{itemize}
    \item [(1)] $\boldsymbol{\beta}^{-1}$ is conditionally negative semidefinite;
    \item [(2)] $\boldsymbol{\sigma} \, \boldsymbol{\beta}^{-d/2}$ is positive semidefinite.
\end{itemize}
\end{lem}

\begin{proof}
The matrix-valued spectral density of the $p$-variate Gaussian covariance \eqref{multi_gauss} in $\mathbb{R}^d$ is (Eq. \eqref{spectraldensitygauss})
\begin{equation}
\label{demogauss}
\widetilde{{\bf G}}(\boldsymbol{\omega};\boldsymbol{\beta},\boldsymbol{\sigma}) = \frac{1}{(4\pi)^{d/2}} \left[\boldsymbol{\sigma} \, \boldsymbol{\beta}^{-d/2}\right] \, \left[\exp\left(-\frac{\| \boldsymbol{\omega} \| ^2}{4\boldsymbol{\beta}}\right)\right], \quad \boldsymbol{\omega} \in \mathbb{R}^d.
\end{equation}
Condition (1) and Lemma~\ref{superlemma} ensure that the last matrix into brackets in the right-hand side member of \eqref{demogauss} is positive semidefinite, while Condition (2) ensures the positive semidefiniteness of the first matrix into brackets. Accordingly, based on Schur's product theorem, the spectral density matrix $\widetilde{{\bf G}}(\boldsymbol{\omega};\boldsymbol{\beta},\boldsymbol{\sigma})$ is positive semidefinite for any frequency vector $\boldsymbol{\omega}$ in $\mathbb{R}^d$, which in turn ensures that the associated covariance model is valid.
\end{proof}

\begin{remark}
A necessary condition for the element-wise inverse matrix $\boldsymbol{\beta}^{-1}$ to be conditionally negative semidefinite (Condition (1) in Lemma~\ref{validityGaussian}) is that $\boldsymbol{\beta}$ is positive semidefinite (Lemma~\ref{lemma2}). Likewise, a necessary condition for Condition (2) to hold is that $\boldsymbol{\sigma}$ is positive semidefinite, insofar as it corresponds to the covariance matrix between the components of the $p$-variate random field at the same spatial location.
\end{remark}

\section{Transitive Upgrading (Mont\'ee) of a Covariance Function}
\label{montee}

The transitive upgrading, also known as ``mont\'ee'' or Radon transform \citep{Matheron1965}, of order $1$ of an absolutely integrable real-valued function $\varphi_d$ defined in $\mathbb{R}^d$ is the function $\varphi_{d-1}$ of $\mathbb{R}^{d-1}$ obtained by integration of $\varphi_d$ along the last coordinate axis:
\begin{equation}
\varphi_{d-1}(s_1,\cdots,s_{d-1}) = \int_{-\infty}^{+\infty} \varphi_{d}(s_1,\cdots,s_{d-1},s_d) \text{d} s_d, \quad (s_1,\cdots,s_{d-1}) \in \mathbb{R}^{d-1}.\\
\end{equation}

The Fourier transforms of $\varphi_d$ and its mont\'ee are related by
\begin{equation}
\widetilde{\varphi}_{d-1}(\omega_1,\cdots,\omega_{d-1})=\widetilde{\varphi}_d(\omega_1,\cdots,\omega_{d-1},0), \quad (\omega_1,\cdots,\omega_{d-1}) \in \mathbb{R}^{d-1}.
\end{equation}

The same relation applies to the Fourier transforms of continuous and absolutely integrable covariance functions: the mont\'ee operator consists in cancelling out the last coordinate of the Fourier transform (spectral density) of the $d$-dimensional covariance to find out that of the $(d-1)$-dimensional covariance. In the isotropic setting where the covariance and its spectral density are radial functions, the mont\'ee leaves unchanged the radial part (a function of the distance to the origin), {\em i.e.}, the mont\'ee of an isotropic covariance defined in $\mathbb{R}^d$ is an isotropic covariance in $\mathbb{R}^{d-1}$ whose spectral density has the same radial part as that of the original covariance \citep{Matheron1965}. For example, the mont\'ee of the isotropic Mat\'ern covariance $k(\cdot ; \alpha,\nu)$ as in Equation \eqref{matern} in $\mathbb{R}^d$ is, up to a multiplicative factor, the Mat\'ern covariance $k(\cdot ; \alpha,\nu+1/2)$ in $\mathbb{R}^{d-1}$.

By repeated application of the mont\'ee of order $1$, one obtains a mont\'ee of any integer order. In particular, the mont\'ee of integer order $\vartheta$ of the isotropic Mat\'ern covariance $k(\cdot ; \alpha,\nu)$ in $\mathbb{R}^d$ is proportional to the Mat\'ern covariance $k(\cdot ; \alpha,\nu+\vartheta/2)$ in $\mathbb{R}^{d-\vartheta}$, the spectral densities of both Mat\'ern covariances having, up to a multiplicative factor, the same radial parts (Eq. \eqref{spectraldensitymatern}). By using the formalism of Hankel transforms instead of that of Fourier transforms, it is possible to generalize the mont\'ee to fractional orders \citep{Matheron1965}.

\section{Proofs}
\label{app3}

\begin{proof}[Proof of Theorem~\ref{validityexponentialdD}]
Owing to Corollary~\ref{corollary3}, it suffices to establish the result for $\nu$ being an integer or a half-integer. We propose a proof based on the representation of the exponential and Mat\'ern covariance as scale mixtures of compactly-supported Askey and Wendland covariances, respectively. The Wendland covariance with range $a>0$ and smoothness parameter $m \in \mathbb{N}$ and its spectral density in $\mathbb{R}^d$ are \citep{Chernih}:
\begin{equation}
\label{wendland1}
w({\boldsymbol{h}};a,m,d) = \frac{2^{-\nu-m}\Gamma(\mu+1)}{ \Gamma(\mu+m+1)}\left(1-\frac{\|\mathbf{h}\|^2}{a^2}\right)_{+}^{\mu+m} \, {}_2F_1\left(\frac{\mu}{2},\frac{\mu+1}{2};\mu+m+1;1-\frac{\|\mathbf{h}\|^2}{a^2}\right),  \quad \boldsymbol{h} \in \mathbb{R}^d
\end{equation}
\begin{equation}
\label{wendland2}
\begin{split}
 \widetilde{w} & (\boldsymbol{\omega}; a,m,d) =  \frac{2^m \, a^{d} \Gamma(\mu+1)\Gamma\left(\frac{1+d+2m}{2}\right)}{\pi^{(d+1)/2}\Gamma(\mu+d+2m+1)} \\
 & \times {}_1F_2\left(\frac{1+d+2m}{2}; \frac{1+d+2m+\mu}{2}, 1+\frac{d+2m+\mu}{2}; -\left(\frac{a\|\boldsymbol{\omega}\|}{2}\right)^2\right) , \quad \boldsymbol{\omega} \in \mathbb{R}^{d},
\end{split}
\end{equation}
with $\mu=
\left\lfloor \frac{d}{2}+m+1 \right\rfloor$, $(\cdot)_+$ the positive part function, ${}_2F_1$ the Gauss hypergeometric function, and ${}_1F_2$ a generalized hypergeometric function. The Askey covariance corresponds to the particular case when $m=0$.

Suppose now that $\nu$ is an integer or a half-integer and let us decompose the univariate exponential covariance in $\mathbb{R}^{d+2\nu-1}$, $d \geq 1$, as a scale mixture of Askey covariances in $\mathbb{R}^{d+2\nu-1}$ with exponent $\mu=\left\lfloor \frac{d+1}{2}+\nu \right\rfloor$:
\begin{equation}
\label{wendland3}
    k({\boldsymbol{h}};\alpha,1/2) = \int_{0}^{+\infty} w({\boldsymbol{h}};t,0,d+2\nu-1) \phi(t; \alpha,\nu,d) \text{d}t, \quad \boldsymbol{h} \in \mathbb{R}^{d+2\nu-1},
\end{equation}
that is,
\begin{equation}
\label{wendland4}
    \exp(-\alpha \|\boldsymbol{h}\|) = \int_{\|\boldsymbol{h}\|}^{+\infty} \left( 1 - \frac{\|\boldsymbol{h}\|}{t}\right)^\mu \phi(t;\alpha,\nu,d) \text{d}t, \quad \boldsymbol{h} \in \mathbb{R}^{d+2\nu-1}.
\end{equation}
To determine $\phi$, we differentiate $(\mu+1)$ times under the integral sign, which is permissible owing to the dominated convergence theorem, to obtain
\begin{equation}
\label{wendland5}
    \phi(t;\alpha,\nu,d) = \frac{\alpha^{\mu+1} \, t^\mu}{\Gamma(\mu+1)} \exp(-\alpha \, t), \quad t>0.
\end{equation}
One recognizes the gamma probability density with shape parameter $\mu+1$ and rate parameter $\alpha$, a result that reminds of the representation in $\mathbb{R}^3$ of the exponential covariance as a gamma mixture of spherical covariances \citep{EmeryLantu2006}. In terms of spectral density, \eqref{wendland5} translates into
\begin{equation}
\label{wendland6}
    \widetilde{k}(\boldsymbol{\omega}; \alpha,1/2) = \int_{0}^{+\infty} \widetilde{w}(\boldsymbol{\omega}; t,0,d+2\nu-1) \, \phi(t;\alpha,\nu,d) \text{d}t, \quad \boldsymbol{\omega} \in \mathbb{R}^{d+2\nu-1}.
\end{equation}

To generalize these results to the Mat\'ern covariance of integer or half-integer parameter $\nu$, let us consider a transitive upgrading (``mont\'ee'') of order $2\nu-1$, which provides a covariance model with the same radial spectral density in a space whose dimension is reduced by $2\nu-1$ (Appendix~\ref{montee}), {\em i.e.}, in $\mathbb{R}^d$. In particular, up to a multiplicative constant, the upgrading of $k(\cdot;\alpha,1/2)$ and $w(\cdot;t,0,d+2\nu-1)$, both covariances being defined in $\mathbb{R}^{d+2\nu-1}$, yields $k(\cdot;\alpha,\nu)$ and $w(\cdot;t,\nu-1/2,d)$, both defined in $\mathbb{R}^{d}$, respectively. Based on \eqref{wendland6}, the upgraded spectral densities in $\mathbb{R}^{d}$ are found to be related by
\begin{equation}
\label{wendland7}
    \frac{\Gamma(\nu)}{\alpha^{2\nu-1} \, \Gamma(1/2)} \widetilde{k}(\boldsymbol{\omega}; \alpha,\nu) = \int_{0}^{+\infty} \frac{t^{2\nu-1}}{2^{\nu-1/2}} \, \widetilde{w}(\boldsymbol{\omega}; t,\nu-1/2,d) \phi(t;\alpha,\nu,d) \text{d}t, \quad \boldsymbol{\omega} \in \mathbb{R}^{d},
\end{equation}
and the covariance functions therefore satisfy the following identity:
\begin{equation}
\label{wendland8}
    \frac{\Gamma(\nu)}{\alpha^{2\nu-1} \, \Gamma(1/2)} {k}(\boldsymbol{h}; \alpha,\nu) = \int_{0}^{+\infty} \frac{t^{2\nu-1}}{2^{\nu-1/2}} \, {w}(\boldsymbol{h}; t,\nu-1/2,d) \phi(t;\alpha,\nu,d) \text{d}t, \quad \boldsymbol{h} \in \mathbb{R}^{d}.\\
\end{equation}
Based on this scale mixture representation, the multivariate Matérn covariance in $\mathbb{R}^d$ with parameters $\boldsymbol{\alpha}$, $\nu \boldsymbol{1}$ and $\boldsymbol{\sigma}$ can be written as
\begin{equation}
\label{wendland9}
{\bf K}({\boldsymbol{h}};\boldsymbol{\alpha},\nu \boldsymbol{1},\boldsymbol{\sigma}) = \boldsymbol{\sigma} \, \int_{0}^{+\infty} \frac{\Gamma(1/2) (\boldsymbol{\alpha}\,t)^{2\nu+\mu} }{{2}^{\nu-1/2} \, \Gamma(\nu)\,\Gamma(\mu+1)\, t} \, \exp(-\boldsymbol{\alpha} t) \, w(\boldsymbol{h};t,\nu-1/2,d) \text{d}t, \quad \boldsymbol{h} \in \mathbb{R}^{d},
\end{equation}
where the products, exponent, exponential and integration are taken element-wise. Under Condition (A.1) of Theorem~\ref{validityexponentialdD}, $\exp(-t \boldsymbol{\alpha})$ is positive semidefinite (Lemma~\ref{superlemma}). If, furthermore, Condition (A.2) also holds, then $\boldsymbol{\sigma} \, \boldsymbol{\alpha}^{2\nu+\mu} \, \exp(-\boldsymbol{\alpha} t)$ is positive semidefinite for any $t>0$, as the element-wise product of positive semidefinite matrices. The matrix-valued Mat\'ern covariance function ${\bf K}({\boldsymbol{h}};\boldsymbol{\alpha},\nu \boldsymbol{1},\boldsymbol{\sigma})$ appears as a mixture of valid Wendland covariances in $\mathbb{R}^{d}$ weighted by positive semidefinite matrices, thus it is a valid model in $\mathbb{R}^{d}$, which completes the proof.
\end{proof}

\begin{proof}[Proof of Theorem~\ref{validitymatern20}]
We first prove (A).  The matrix-valued spectral density of the multivariate Mat\'ern covariance ${\bf K}({\boldsymbol{h}}; \boldsymbol{\alpha}, \boldsymbol{\nu},\boldsymbol{\sigma})$ in $\mathbb{R}^d$ is (Eq. \eqref{spectraldensitymatern})
\begin{equation}
\label{demospecmatern4}
\widetilde{{\bf K}}(\boldsymbol{\omega}; \boldsymbol{\alpha}, \boldsymbol{\nu},\boldsymbol{\sigma}) = \left[ \frac{{\boldsymbol{\sigma}} \, \Gamma(\boldsymbol{\nu}+d/2)}{\pi^{d/2} \,\boldsymbol{\alpha}^{d} \, \Gamma(\boldsymbol{\nu})} \right] \, \left[\exp\left(-\left(\boldsymbol{\nu}+\frac{d}{2}\right) \ln \left(1+\frac{\| \boldsymbol{\omega} \| ^2}{\boldsymbol{\alpha}^2}\right)\right)\right],
\quad \boldsymbol{\omega} \in \mathbb{R}^d.
\end{equation}

Since the composition of two Bernstein functions is still a Bernstein function and based on the fact that $x \mapsto \ln(1+x)$ is a Bernstein function \citep[corollary 3.8 and chapter 16.4]{Schilling}, the matrix
\begin{equation*}
    \left[\left(\boldsymbol{\nu}+\frac{d}{2}\right) \ln \left(1+\frac{\| \boldsymbol{\omega} \| ^2}{\boldsymbol{\alpha}^2}\right)\right]
\end{equation*}
turns out to be the element-wise product of two Bernstein matrices with the same supporting points under Conditions (A.1) and (A.2) of Theorem~\ref{validitymatern20}, hence it is conditionally negative semidefinite (see example following \eqref{vario}). If Condition (A.3) also holds, then the spectral density matrix $\widetilde{{\bf K}}(\boldsymbol{\omega}; \boldsymbol{\alpha}, \boldsymbol{\nu},\boldsymbol{\sigma})$ is positive semidefinite for any $\boldsymbol{\omega} \in \mathbb{R}^d$, based on Lemma~\ref{superlemma} and Schur's product theorem, and ${\bf K}(\boldsymbol{h}; \boldsymbol{\alpha}, \boldsymbol{\nu},\boldsymbol{\sigma})$ is therefore a valid covariance function in $\mathbb{R}^d$.  \\\\
We now prove (B).
Consider the $p$-variate Mat\'ern covariance \eqref{multi_cov}, in which each entry is a scale mixture of Gaussian covariances of the form \eqref{integral_repr}. Let $\boldsymbol{\psi}=[\psi_{ij}]_{i,j=1}^p$ be a matrix with positive entries. The change of variable $u = \psi^{-1}_{ij} v$ in the scale mixture representation of the cross-covariance $k_{ij}(\boldsymbol{h})$ yields the following expression for the $p$-variate Mat\'ern covariance model, where the exponent $-1$ indicates the element-wise inverse:
\begin{equation}
 \label{integral_repr3}
{\bf K}({\boldsymbol{h}};\boldsymbol{\alpha},\boldsymbol{\nu},\boldsymbol{\sigma}) = \int_{0}^{+\infty}   g(\boldsymbol{h}; \boldsymbol{\psi}^{-1} v) \, \boldsymbol{\sigma} \, {\bf F}(\boldsymbol{\psi}^{-1} v; \boldsymbol{\alpha},\boldsymbol{\nu}) \, \boldsymbol{\psi}^{-1}  \text{d}v.
\end{equation}
Based on Lemma~\ref{validityGaussian}, sufficient conditions for this model to be valid in $\mathbb{R}^d$ are
\begin{itemize}
    \item [(1)] $\boldsymbol{\psi}$ is conditionally negative semidefinite;
    \item [(2)] $\boldsymbol{R}(v)= \boldsymbol{\sigma} {\bf F}(\boldsymbol{\psi}^{-1} v; \boldsymbol{\alpha}, \boldsymbol{\nu})  \boldsymbol{\psi}^{-1} (\boldsymbol{\psi}^{-1} v)^{-d/2}$ is positive semidefinite for any $v \geq 0$.
\end{itemize}
The latter matrix can be rewritten as
\begin{align*}
\begin{split}
\boldsymbol{R}(v) & = \frac{1}{v^{1+d/2}} \,  \frac{\boldsymbol{\sigma} \boldsymbol{\psi}^{d/2}}{ \Gamma(\boldsymbol{\nu})} \left( \frac{\boldsymbol{\alpha}^2 \, \boldsymbol{\psi}}{4 v} \right)^{\boldsymbol{\nu}}  \exp\left(-\frac{\boldsymbol{\alpha}^2\, \boldsymbol{\psi}}{4 v}\right) \\
&=  (4t(v))^{1+d/2} \,  \frac{\boldsymbol{\sigma} \boldsymbol{\psi}^{d/2}}{ \Gamma(\boldsymbol{\nu})} \left( t(v) \boldsymbol{\alpha}^2\, \boldsymbol{\psi} \right)^{\boldsymbol{\nu}}  \exp\left(-t(v) \boldsymbol{\alpha}^2\, \boldsymbol{\psi}\right),
\end{split}
\end{align*}
with $t(v) = \frac{1}{4v} > 0$. To prove that, under Conditions (B.2) to (B.4) of Theorem~\ref{validitymatern20}, this matrix is positive semidefinite for any positive value of $t(v)$, we distinguish two cases.\\
\begin{itemize}
    \item $t(v) \leq 1$.
Up to a positive scalar factor, $\boldsymbol{R}(v)$ can be decomposed as follows:
\begin{equation}
\label{tlessthan1}
\begin{split}
 \boldsymbol{R}(v) &=  (4t(v))^{1+d/2} \left[\exp((\ln t(v)+1-t(v)) \boldsymbol{\nu})\right] \\ & \quad \left[\exp\left(-t(v) \left(\boldsymbol{\alpha}^2\, \boldsymbol{\psi}-\boldsymbol{\nu}\right)\right) \right]
 \, \left[ \frac{\boldsymbol{\sigma} \boldsymbol{\psi}^{d/2}}{\Gamma(\boldsymbol{\nu})} \left(\boldsymbol{\alpha}^{2}\, \boldsymbol{\psi}\right)^{\boldsymbol{\nu}} \exp(-\boldsymbol{\nu}) \right],
\end{split}
\end{equation}
with $\ln t(v) + 1- t(v) < 0$ and $-t(v) < 0$. Together with Lemma~\ref{superlemma}, Conditions (B.2), (B.3) and (B.4) of Theorem~\ref{validitymatern20} ensure the positive semidefiniteness of the first, second and third matrices into brackets in the second member of \eqref{tlessthan1}, respectively. Accordingly, $\boldsymbol{R}(v)$ is positive semidefinite based on Schur's product theorem.\\

    \item $t(v) > 1$. One has the following decomposition:
\begin{equation}
\label{tmorethan1}
\begin{split}
  \boldsymbol{R}(v) =  &(4t(v))^{1+d/2} \, \left[\exp\left(-\left(\ln t(v) + 1\right) \, \left(\boldsymbol{\alpha}^2\, \boldsymbol{\psi}- {\boldsymbol{\nu}}\right)\right)\right] \\
    & \, \left[\exp \left(\left(\ln t(v) +1- t(v)\right) \, \boldsymbol{\alpha}^2\, \boldsymbol{\psi}\right)\right] \, \left[ \frac{\boldsymbol{\sigma} \boldsymbol{\psi}^{d/2}}{\Gamma(\boldsymbol{\nu})} \left(\boldsymbol{\alpha}^{2}\, \boldsymbol{\psi}\right)^{\boldsymbol{\nu}} \exp(-\boldsymbol{\nu}) \right],
\end{split}
\end{equation}
with $-(\ln t(v) + 1)<0$ and $\ln t(v)+ 1- t(v)<0$. The positive semidefiniteness of $\boldsymbol{R}(v)$ follows from Conditions (B.2), (B.3) and (B.4) of Theorem~\ref{validitymatern20}, Lemma~\ref{superlemma} and Schur's product theorem. In particular, one uses the fact that Conditions (B.2) and (B.3) imply the conditional negative semidefiniteness of $\boldsymbol{\alpha}^2\,\boldsymbol{\psi}$.\end{itemize}
\end{proof}

\begin{proof}[Proof of Theorem~\ref{validityMatern}]
Conditions (A) and (B) are particular cases of Theorem~\ref{validitymatern20}B, with $\boldsymbol{\psi}=\boldsymbol{\nu} \, \boldsymbol{\alpha}^{-2}$ and $\boldsymbol{\psi}=\boldsymbol{1}/\beta$, respectively.

\end{proof}

\renewcommand*{\thesection}{\Alph{section}}


\bibliographystyle{apalike}
\bibliography{mybib}

\end{document}